\documentclass[12pt,onecolumn, draftclsnofoot]{IEEEtran}
\usepackage{amssymb}
\usepackage{graphicx}
\usepackage{epstopdf}
\usepackage[tbtags]{amsmath}
\usepackage{amsfonts}
\usepackage{mathtools}
\usepackage{mathrsfs}
\usepackage{amsthm}
\usepackage{xcolor}
\usepackage{cite,url}
\usepackage{upquote}
\usepackage{pifont}
\usepackage{amsmath}
\usepackage{mathtools}

\usepackage{units}
\usepackage{booktabs}
\usepackage{multirow}
\usepackage{array}

\newtheorem{thm}{Theorem}
\newtheorem{theorem}[thm]{Theorem}
\newtheorem{lemma}[thm]{Lemma}
\newtheorem{definition}[thm]{Definition}

\newtheorem{corollary}[thm]{Corollary}

\newtheorem{example}{Example}

\newtheorem{remark}{Remark}
\begin{document}
\title{Quantum Relative $\alpha$-Entropies: A Structural and Geometric Perspective}


 \author{%
   \IEEEauthorblockN{Sayantan Roy\IEEEauthorrefmark{1},
                    Atin Gayen\IEEEauthorrefmark{2},
                    Aditi Kar Gangopadhyay\IEEEauthorrefmark{3},
                   and Sugata Gangopadhyay\IEEEauthorrefmark{4}}
                 \vspace{0.4cm}
                 
  \IEEEauthorblockA{\IEEEauthorrefmark{1}%
                    \small{Department of Mathematics,
                   Indian Institute of Technology Roorkee, Roorkee, 247677, India,
                    roysayantan3024@gmail.com}}
                    \\
   \IEEEauthorblockA{\IEEEauthorrefmark{2}%
                   \small{Department of Computer Science and Engineering,
                   Indian Institute of Technology Roorkee, Roorkee, 247677, India,
                     atinfordst@gmail.com}}\\
   \IEEEauthorblockA{\IEEEauthorrefmark{3}%
                     \small{Department of Mathematics,
                   Indian Institute of Technology Roorkee, Roorkee, 247677, India,
            aditi.gangopadhyay@ma.iitr.ac.in}}\\
   \IEEEauthorblockA{\IEEEauthorrefmark{4}%
                     \small{Department of Computer Science and Engineering,
                   Indian Institute of Technology Roorkee, Roorkee, 247677, India,
            sugata.gangopadhyay@cs.iitr.ac.in}}
}

\maketitle


\begin{abstract}
Most quantum divergences derive their structure from classical $f$-divergences or R\'enyi-type constructions, a dependence that obscures several quantum geometric effects.
We introduce a quantum relative-$\alpha$-entropy that extends Umegaki’s relative entropy while falling outside the $f$-divergence class. The proposed divergence exhibits a nonlinear convexity property, which yields a generalized convexity result for the Petz-R\'enyi divergence for $\alpha>1$, complementing the known convexity for $\alpha<1$. It is additive under tensor products, invariant under unitary transformations, and depends only on the relative geometry of quantum states rather than their absolute magnitudes. Using Nussbaum-Szko\l{}a-type distributions, we also establish an exact correspondence of this divergence with classical relative-$\alpha$-entropy. This reveals relative-$\alpha$-entropy as a fundamentally geometric notion of quantum distinguishability not captured by existing divergence frameworks.
\end{abstract}



\section{Introduction}
\label{2sec:introduction}
Quantum information divergences quantify the distinguishability between quantum states and play a central role across quantum information theory, including quantum cryptography, computation, and learning \cite{Marcol2015}. Conceptually, they extend classical information-theoretic distances to the non-commutative setting of quantum mechanics.

Let $\mathcal{H}$ be an $n$-dimensional complex Hilbert space. A quantum state on $\mathcal{H}$ is described by a density matrix $\rho$, that is, a positive semi-definite, Hermitian operator with unit trace. Pure states correspond to rank-one projections, while mixed states represent statistical ensembles of pure states. By the spectral theorem, any density matrix admits the decomposition
\begin{equation}
\label{rho}
\rho = \sum_{i=1}^n \lambda_i \, |x_i\rangle \langle x_i|,
\end{equation}
where $\{\lambda_i\}_{i=1}^n$ are non-negative eigenvalues summing to one, and $\{|x_i\rangle\}_{i=1}^n$ form an orthonormal basis of $\mathcal{H}$. The set of all density matrices on $\mathcal{H}$ is convex, with pure states as its extreme points \cite{Denes2007}.

Quantum divergences provide quantitative measures of distance between quantum states. While any operator norm induces a notion of distance, information-theoretic applications require divergences that reflect statistical and operational structure \cite{Nielsen00}. The most prominent example is Umegaki’s quantum relative entropy \cite{Umegaki62}, defined for density operators $\rho$ and $\sigma$ by
\begin{equation}
    \label{Umegaki}
    U(\rho \| \sigma) = \operatorname{Tr}(\rho \log \rho - \rho \log \sigma).
\end{equation}
This quantity is nonnegative and vanishes if and only if $\rho=\sigma$. Further, it is finite whenever $\operatorname{supp}(\rho) \subseteq \operatorname{supp}(\sigma)$. Here supp($\rho$) denotes the support of any density matrix $\rho$, which is defined to be the span of the eigenvectors of $\rho$ corresponding to its non-zero eigenvalues. Any vector orthogonal to supp($\rho$) lies within the kernel of $\rho$, hence
$U(\rho||\sigma) = +\infty$ when $\operatorname{supp}(\rho)\cap \operatorname{ker}(\sigma) \neq\{0\}$. 

Quantum relative entropy plays a central role in quantum information theory and quantum statistical learning. Its operational significance was established in quantum hypothesis testing \cite{OgawaNagaoka00}, and it underlies several fundamental information measures, including the von Neumann entropy
\begin{equation}
    \label{vonNeumann}
    S(\rho) = - \operatorname{Tr}(\rho \log \rho),
\end{equation}
as well as conditional entropy and coherent information \cite{Ohya2010,Bluhm2022}. For example,
\[
U(\rho \| \sigma) = C(\rho,\sigma) -S(\rho),
\]
where $C(\rho,\sigma)$ is known as the cross entropy that satisfies $C(\rho,\rho)= S(\rho)$. In this sense, Umegaki’s relative entropy serves as a structural cornerstone of quantum information theory.

Consider a quantum state $\rho$ with spectral decomposition
$\rho = \sum_{i=1}^n \lambda_i |x_i\rangle\langle x_i|$.
Its von Neumann entropy admits the eigenvalue representation
\begin{equation}
    S(\rho) = -\sum_{i=1}^n \lambda_i \log \lambda_i,
\end{equation}
which is formally analogous to the Shannon entropy
\begin{equation}
    H(p) = -\sum_{x\in\mathbb{S}} p(x)\log p(x)
\end{equation}
of a classical probability distribution $p$ supported on a finite set
$\mathbb{S}\subseteq\mathbb{R}$.

A similar correspondence holds for Umegaki's relative entropy.
When $\rho$ and $\sigma$ commute, they are simultaneously diagonalizable,
and \eqref{Umegaki} reduces to the classical Kullback-Leibler (KL) divergence
between their eigenvalue distributions:
\begin{equation}
\label{KL}
U(\rho\|\sigma)
= \sum_{i=1}^n \lambda_i
\log\!\left(\frac{\lambda_i}{\delta_i}\right)
=: \mathrm{KL}(\lambda\|\delta),
\end{equation}
where $\{\lambda_i\}$ and $\{\delta_i\}$ denote the eigenvalues of $\rho$
and $\sigma$, respectively \cite{Kullback51,CsiszarS04B}.
This reduction follows directly from the spectral (Schmidt) decompositions
of the commuting density matrices \cite{zhang11}.

In classical information theory and statistics, several generalizations of the
KL-divergence have been proposed to address problems arising
in noisy communication channels, hypothesis testing, and robust inference \cite{Pardo06B, BasuSP11B}.
One of the most prominent among these is the R\'enyi divergence, which replaces
the logarithmic structure of \eqref{KL} with power functions
\cite{BasuBC97J,Ervan2014}.
Let $p$ and $q$ be two probability distributions with common support
$\mathbb{S}\subseteq\mathbb{R}$, and let $\alpha>0$, $\alpha\neq 1$.
The R\'enyi divergence of order $\alpha$ is defined as \cite{Renyi61J}
\begin{equation}
\label{RenyiClassical}
R_\alpha(p\|q)
:= \frac{1}{\alpha-1}
\log \sum_{x\in\mathbb{S}} p(x)^{\alpha} q(x)^{1-\alpha}.
\end{equation}
It is well known that $R_\alpha(p\|q)\to \mathrm{KL}(p\|q)$ as $\alpha\to 1$.
Moreover, R\'enyi divergence belongs to the broader class of
Csisz\'ar $f$-divergences (as a monotone function of $\alpha$-Hellinger divergence) \cite{Csiszar1967, GayenK21J, BasuSP11B, CressieR84J}.

Despite its wide applicability, the use of R\'enyi divergence is often
technically challenging due to the presence of nonlinear powers in both
arguments, particularly in inference and optimization problems
\cite{GayenK21J,KumarS15J2,GayenRG24J}.
This has motivated the introduction of an alternative generalization,
known as the relative-$\alpha$-entropy
\cite{Sundaresan2002,Sundaresan2007,KumarS15J1,KumarS15J2},
defined as
\begin{equation}
\label{Jones}
\begin{aligned}
J_{\alpha}(p\|q)
&= \frac{\alpha}{1-\alpha}
\log \sum_{x\in\mathbb{S}} p(x) q(x)^{\alpha-1}
- \frac{1}{1-\alpha}
\log \sum_{x\in\mathbb{S}} p(x)^{\alpha}  + \log \sum_{x\in\mathbb{S}} q(x)^{\alpha}.
\end{aligned}
\end{equation}
Unlike R\'enyi divergence, the relative-$\alpha$-entropy involves a linear power
of the first argument (in the product term), however, coincides with the KL divergence in the limit
$\alpha\to 1$.
Furthermore, relative-$\alpha$-entropy is closely related to R\'enyi divergence
through the so-called $\alpha$-escort transformation
$p \mapsto p^{(\alpha)}$, where
$p^{(\alpha)}(x) := p(x)^{\alpha} \big/ \sum_{y\in\mathbb{S}} p(y)^{\alpha}$
\cite{KumarS15J1}.

Our main contributions are summarized as follows.
\vspace{0.1cm}

\begin{itemize}

\item[1.] \textbf{Quantum Relative-$\alpha$-Entropy Beyond the $f$-Divergence Framework:}
We introduce a new quantum divergence, termed the \emph{quantum relative-$\alpha$-entropy}, which constitutes a new generalization of Umegaki’s relative entropy \eqref{Umegaki}. The proposed divergence is parameterized by $\alpha>0$, $\alpha\neq1$, and recovers Umegaki’s relative entropy in the limit $\alpha\to1$. Unlike standard R\'enyi-type constructions, it lies strictly outside the class of quantum $f$-divergences, while retaining several fundamental structural properties.

\item[2.] \textbf{Nonlinear Generalized Convexity and Geometric Structure:}
We show that the quantum relative-$\alpha$-entropy fails to be jointly convex in the usual linear sense. Motivated by its intrinsic multiplicative structure, we introduce a notion of nonlinear generalized convexity adapted to the geometry of the divergence. Within this framework, we establish convexity properties that are not captured by classical formulations. As an application, we derive a generalized convexity result for the Petz-R\'enyi relative entropy in the regime $\alpha>1$, complementing the well-known linear convexity for $\alpha<1$.

\item[3.] \textbf{Structural and Operational Distinctions from R\'enyi-Type Divergences:}
We investigate the relationship between the proposed divergence and existing R\'enyi-type quantum divergences, including the Petz-R\'enyi and sandwiched R\'enyi divergences. We identify several fundamental differences, particularly with respect to monotonicity properties and the data-processing inequality. Through explicit examples, we demonstrate that the quantum relative-$\alpha$-entropy exhibits behavior that is qualitatively distinct from standard R\'enyi-type divergences.

\item[4.] \textbf{Classical-Quantum Correspondence via Nussbaum-Szko\l{}a Distributions:}
We establish a precise correspondence between the quantum relative-$\alpha$-entropy and the classical relative-$\alpha$-entropy by means of Nussbaum-Szko\l{}a distributions associated with a pair of quantum states. This result provides an exact reduction of the quantum divergence to its classical counterpart, thereby offering a unified geometric perspective on classical and quantum notions of distinguishability.

\item[5.] \textbf{A Quantum Bregman-Type Density Power Divergence and Structural Comparison:}
Motivated by the classical density power divergence, we introduce a quantum divergence of Bregman type in the operator setting. This construction can be viewed as a log-free counterpart of the quantum relative-$\alpha$-entropy, obtained by removing the outer logarithmic transformation from its defining expression. We analyze its structural properties and compare it systematically with the quantum relative-$\alpha$-entropy. While the two divergences share a common algebraic backbone, they exhibit distinct geometric and monotonicity behaviors, thereby highlighting the structural role played by the logarithmic transformation in quantum divergence theory.
\end{itemize}

The remainder of the paper is organized as follows. 
Section~\ref{3sec:Motivation} reviews R\'enyi-type quantum divergences and quantum $f$-divergences from the literature. 
In Section~\ref{4sec:QuantumRelativeEntropy}, we introduce the quantum relative-$\alpha$-entropy and establish its fundamental properties. 
Section~\ref{5sec:Connectionsdiscussions} examines its connections with existing quantum information measures and highlights key structural distinctions. 
In Section~\ref{5sec:NZDistributions}, we establish the exact correspondence between the quantum and classical relative-$\alpha$-entropies via Nussbaum-Szko\l{}a distributions. 
Section~\ref{6sec:QuantumDPD} introduces a log-free quantum density power divergence inspired by classical analogues and compares its properties with those of the quantum relative-$\alpha$-entropy. 
Finally, Section~\ref{7sec:summary} concludes the paper with a summary and discussion of future directions.

\section{Background of the Problem}
\label{3sec:Motivation}

In this section, we review some well-known generalizations of quantum relative entropy relevant to the present work, with an emphasis on their structural and operational features.

Umegaki's relative entropy \eqref{Umegaki} was introduced as the quantum analogue of the KL divergence \eqref{KL} in \cite{Kullback51}. Araki subsequently provided a formulation within the framework of relative modular operators \cite{Araki1975, Araki2005}, placing the definition on firm operator-algebraic foundations. In analogy with R\'enyi's extension of the classical KL divergence \cite{Renyi61J}, Petz and Ohya generalized Araki's construction to obtain the Petz--R\'enyi-$\alpha$ relative entropy \cite{Petz1986}. For two quantum states $\rho$ and $\sigma$, it is defined by
\begin{equation}
\label{RenyiRelativeentropy}
\Hat{D}_{\alpha}(\rho\|\sigma)
=
\frac{1}{\alpha - 1}
\log \operatorname{Tr}(\rho^{\alpha}\sigma^{1-\alpha}),
\end{equation}
for $\alpha > 0$, $\alpha \neq 1$.

A further modification, known as the sandwiched R\'enyi divergence, was introduced in \cite{MullerMartin2013} and is given by
\begin{equation}
\label{Sandwiched}
D_{\alpha}^{*}(\rho\|\sigma)
=
\frac{1}{\alpha - 1}
\log
\operatorname{Tr}
\left[
\left(
\sigma^{\frac{1-\alpha}{2\alpha}}
\rho
\sigma^{\frac{1-\alpha}{2\alpha}}
\right)^{\alpha}
\right].
\end{equation}
Both divergences admit well-defined extensions to the limiting cases $\alpha \to 0, 1,$ and $+\infty$ \cite{MullerMartin2013, Datta15}. For example,
\begin{itemize}
    \item[1.] when $\alpha\to 1$, they both coincide with Umegaki's relative entropy \eqref{Umegaki}.
    \item[2.] when $\alpha\to \infty$, they coincide with the max-relative entropy $D_{max}(\rho||\sigma)$, defined in \cite{Datta2009} as 
    \begin{equation*}
        D_{max}(\rho||\sigma) := \log \min\{\lambda : \rho \leq \lambda \sigma\}.
    \end{equation*}
\end{itemize}

 They play a central role in quantum information theory and possess distinct operational interpretations, particularly in quantum hypothesis testing and asymptotic error exponents \cite{Nussbaum2009, Hiai2017, Hiai2018}. Importantly, although they coincide in certain parameter regimes, their structural properties, such as monotonicity under quantum channels and convexity behavior, differ in essential ways.

Beyond R\'enyi-type quantities, a broader class of divergences arises from quantum analogues of the classical Csisz\'ar $f$-divergence. Let $f : (0,+\infty) \to \mathbb{R}$ be convex. For two quantum states $\rho$ and $\sigma$, the standard quantum $f$-divergence is defined as
\begin{equation}
\label{standard_f}
S_f(\rho\|\sigma)
=
\langle \rho^{1/2},\, f(\Delta(\sigma,\rho))\, \rho^{1/2} \rangle,
\end{equation}
where $\Delta(\sigma,\rho)$ denotes the relative modular operator and
$\langle \rho,\sigma \rangle = \operatorname{Tr}(\rho^* \sigma)$ with $\rho^*$ being the adjoint of $\rho$. This construction extends the classical Csisz\'ar $f$-divergence
\begin{equation}
\label{f-divergence}
D_f(p\|q)
=
\sum_{x \in \mathbb{S}}
q(x)
f\!\left(\frac{p(x)}{q(x)}\right),
\end{equation}
defined for probability distributions $p$ and $q$ with common support $\mathbb{S}$.

As in the classical setting, the quantum $f$-divergence includes several important divergences as special cases. In particular, Umegaki's relative entropy \eqref{Umegaki} and the Petz-R\'enyi relative entropy \eqref{RenyiRelativeentropy} can be recovered from \eqref{standard_f} for suitable choices of $f$. At the same time, not all quantum divergences, most interestingly the sandwiched R\'enyi divergence, fit into this framework. The coexistence of these inequivalent extensions highlights the structural diversity of quantum relative entropies and motivates further investigation into alternative formulations and their properties. This perspective underlies the developments considered in the present work.

\section{Quantum Relative $\alpha$-Entropy : Definition and Properties}
\label{4sec:QuantumRelativeEntropy}

The Csisz\'ar $f$-divergence and Bregman divergence families represent two central frameworks for classical information divergences, with the KL divergence being an important member of both \cite{KumarM20J}. However, the relative-$\alpha$-entropy $J_\alpha(p||q)$, another generalization of the KL divergence, lies outside these families while still retaining several fundamental properties of interest that are significant in Information Theory and Statistical Learning \cite{Sundaresan2002, Sundaresan2007, JonesHHB01J, KumarS15J2, GayenK21J, GayenK23J}.

Motivated by this structure, in this section, we propose a class of quantum divergences that generalizes Umegaki’s relative entropy \eqref{Umegaki}, and also falls outside the quantum $f$-divergence class \eqref{standard_f}. We establish key properties of this class, including additivity under tensor products, unitary invariance, convexity, and so on. We also list out some of its properties that are unique to this class. We start by recalling some of the necessary basic definitions from Operator Theory \cite{Parthasarathy06, Denes2007, zhang11}.

Let $\mathcal{H}$ be a finite-dimensional Hilbert space with $\dim(\mathcal{H}) = n$, and let $\mathcal{B}(\mathcal{H})$ denote the algebra of all bounded linear operators on $\mathcal{H}$. For any $X \in \mathcal{B}(\mathcal{H})$, its adjoint is denoted by $X^{*}$ and is defined through the relation
\[
\langle u, Xv \rangle = \langle X^{*}u, v \rangle, \qquad u,v \in \mathcal{H}.
\]
The Hilbert-Schmidt inner product on $\mathcal{B}(\mathcal{H})$ is given by
\[
\langle X, Y \rangle = \operatorname{Tr}(X^{*}Y).
\]
For $X \in \mathcal{B}(\mathcal{H})$, we define the Schatten $p$-norm as
\begin{equation}
\label{pnorm1}
\|X\|_{p} := \big( \operatorname{Tr}|X|^{p} \big)^{1/p},
\end{equation}
where $p \geq 1$ and $|X| := (X^{*}X)^{1/2}$. This definition extends naturally to negative values of $p$. Moreover, the case $p = \infty$ is defined by
\[
\|X\|_{\infty} := \lim_{p \to \infty} \|X\|_{p}.
\]
For $1 \leq p \leq \infty$, $\|\cdot\|_{p}$ defines a norm on $\mathcal{B}(\mathcal{H})$ and satisfies the triangle inequality.

In particular, for a density operator $\rho$ with spectral decomposition
\[
\rho = \sum_{i=1}^{n} \lambda_i \, |x_i\rangle \langle x_i|,
\]
the Schatten $p$-norm reduces to
\begin{equation*}
\|\rho\|_{p} = \left( \sum_{i=1}^{n} \lambda_i^{p} \right)^{1/p}.
\end{equation*}

\subsection{Proposal of Quantum Relative-$\alpha$-Entropy}

\begin{definition}
\label{Defn:Divergence}
Let $\alpha > 0$, $\alpha \neq 1$. For two density operators $\rho$ and $\sigma$ acting on a finite-dimensional Hilbert space, the \emph{quantum relative $\alpha$-entropy} is defined as
\begin{equation}
\label{quantumjonesetal}
S_{\alpha}(\rho \| \sigma)
= \frac{\alpha}{1-\alpha} \log \operatorname{Tr}\!\left( \rho \sigma^{\alpha-1} \right)
- \frac{1}{1-\alpha} \log \operatorname{Tr}\!\left( \rho^{\alpha} \right)
+ \log \operatorname{Tr}\!\left( \sigma^{\alpha} \right),
\end{equation}
whenever $\operatorname{supp}(\rho) \subseteq \operatorname{supp}(\sigma)$. Otherwise, we set
\[
S_{\alpha}(\rho \| \sigma) := +\infty.
\]
\end{definition}

Throughout this paper, we adopt the following conventions:
\[
0 \cdot (\pm\infty) = 0, 
\qquad 
\log 0 = -\infty, 
\qquad 
\log(+\infty) = +\infty.
\]

Using the linearity of the trace operator, that is, $\operatorname{Tr}(cA) = c\,\operatorname{Tr}(A)$ for any scalar $c$, the expression in \eqref{quantumjonesetal} can be equivalently rewritten as
\begin{align}
\label{1:quantumjones1}
S_{\alpha}(\rho \| \sigma)
&= \frac{\alpha}{1-\alpha} \log \operatorname{Tr}\!\left( \rho \sigma^{\alpha-1} \right)
   - \frac{1}{1-\alpha} \log \operatorname{Tr}\!\left( \rho^{\alpha} \right)
   + \log \operatorname{Tr}\!\left( \sigma^{\alpha} \right) \nonumber \\
&= \frac{\alpha}{1-\alpha}
   \log \operatorname{Tr}\!\left[
   \rho \sigma^{\alpha-1}
   \left( \operatorname{Tr}\rho^{\alpha} \right)^{-1/\alpha}
   \left( \operatorname{Tr}\sigma^{\alpha} \right)^{-(\alpha-1)/\alpha}
   \right] \nonumber \\
&= \frac{\alpha}{1-\alpha}
   \log \operatorname{Tr}\!\left[
   \frac{\rho}{\|\rho\|_{\alpha}}
   \left( \frac{\sigma}{\|\sigma\|_{\alpha}} \right)^{\alpha-1}
   \right],
\end{align}
where $\|\cdot\|_{\alpha}$ denotes the Schatten $\alpha$-norm \eqref{pnorm1}.

Now we motivate the similarity between the expressions of the proposed quantum relative
$\alpha$-entropy $S_{\alpha}$ and the classical relative-$\alpha$-entropy
$J_{\alpha}$. To this end, we first establish the following lemma.

\begin{lemma}
\label{Lemma:traceequation}
Let $\rho$ and $\sigma$ be two density operators with respective spectral decompositions
\begin{equation}
\label{decompositions}
\rho = \sum_{i=1}^{n} p_i \, |x_i\rangle \langle x_i|,
\qquad
\sigma = \sum_{j=1}^{n} q_j \, |y_j\rangle \langle y_j|,
\end{equation}
where $\sum_{i=1}^{n} p_i = \sum_{j=1}^{n} q_j = 1$ and $p_i,q_j \geq 0$ for all
$i,j$. Then, for $\alpha > 0$,
\begin{equation}
\label{1eq:traceeq}
\operatorname{Tr}\!\left( \rho \sigma^{\alpha-1} \right)
= \sum_{i,j=1}^{n} p_i q_j^{\alpha-1}
\left| \langle x_i | y_j \rangle \right|^2 .
\end{equation}
\end{lemma}

\begin{IEEEproof}
By the spectral theorem,
\[
\rho = \sum_{i=1}^{n} p_i |x_i\rangle \langle x_i|,
\qquad
\sigma^{\alpha-1} = \sum_{j=1}^{n} q_j^{\alpha-1} |y_j\rangle \langle y_j|.
\]
Therefore,
\begin{align*}
\operatorname{Tr}\!\left( \rho \sigma^{\alpha-1} \right)
&= \operatorname{Tr}\!\left(
\sum_{i,j} p_i q_j^{\alpha-1}
|x_i\rangle \langle x_i | y_j\rangle \langle y_j|
\right) \\
&= \sum_{i,j} p_i q_j^{\alpha-1}
\operatorname{Tr}\!\left(
|x_i\rangle \langle x_i | y_j\rangle \langle y_j|
\right).
\end{align*}
Using $\operatorname{Tr}(|u\rangle\langle v|)=\langle v|u\rangle$, we obtain
\[
\operatorname{Tr}\!\left(
|x_i\rangle \langle x_i | y_j\rangle \langle y_j|
\right)
= |\langle x_i | y_j \rangle|^2,
\]
which yields \eqref{1eq:traceeq}.
\end{IEEEproof}

It is immediate from the spectral decomposition that, for a density operator
$\rho$,
\[
\operatorname{Tr}(\rho^{\alpha}) = \sum_{i=1}^{n} p_i^{\alpha}.
\]
Consequently, the quantum relative $\alpha$-entropy defined in
\eqref{quantumjonesetal} admits the equivalent representation
\begin{align}
\label{quantumjones2}
S_{\alpha}(\rho \| \sigma)
&= \frac{\alpha}{1-\alpha}
\log \sum_{i,j=1}^{n}
p_i q_j^{\alpha-1}
|\langle x_i | y_j \rangle|^2
- \frac{1}{1-\alpha}
\log \sum_{i=1}^{n} p_i^{\alpha}
+ \log \sum_{j=1}^{n} q_j^{\alpha}.
\end{align}

\begin{remark}
The representation in \eqref{quantumjones2} highlights the close resemblance
between the quantum relative $\alpha$-entropy $S_{\alpha}(\rho \| \sigma)$ and
the classical relative-$\alpha$-entropy $J_{\alpha}(p \| q)$ defined in
\eqref{Jones}. In particular, when $\rho$ and $\sigma$ correspond to classical
states (that is, they commute and are diagonal in a common basis),
$S_{\alpha}(\rho \| \sigma)$ reduces exactly to $J_{\alpha}(p \| q)$.
\end{remark}

\begin{remark}
\label{Remark:finiteness}
The quantity $S_{\alpha}(\rho \| \sigma)$ is finite for all $\alpha > 0$,
$\alpha \neq 1$, if and only if
\[
\operatorname{supp}(\rho) \cap \operatorname{supp}(\sigma) \neq \emptyset.
\]
Moreover, for $\alpha \leq 1$, $S_{\alpha}(\rho \| \sigma) < \infty$ if and only if
\[
\operatorname{supp}(\rho) \subseteq \operatorname{supp}(\sigma).
\]
\end{remark}

For any density operator $\rho$, we have
$0 < \operatorname{Tr}(\rho^{\alpha}) < \infty$ for all $\alpha > 0$.
Consequently, the finiteness of $S_{\alpha}(\rho \| \sigma)$ hinges on the
positivity and finiteness of the term
$\operatorname{Tr}(\rho \sigma^{\alpha-1})$.
Using the representation derived in \eqref{1eq:traceeq}, we obtain the following
characterization, which underlies Remark~\ref{Remark:finiteness}.

\begin{lemma}
\label{Lemma:supportcharacterization}
Let $\rho$ and $\sigma$ be density operators as in
Lemma~\ref{Lemma:traceequation}, with respective spectral decompositions
\[
\rho = \sum_{i=1}^{n} p_i |x_i\rangle \langle x_i|,
\qquad
\sigma = \sum_{j=1}^{n} q_j |y_j\rangle \langle y_j|.
\]
Then $\operatorname{supp}(\rho) \subseteq \operatorname{supp}(\sigma)$ if and
only if for every $y_j \in \operatorname{Ker}(\sigma)$ and every $x_i \in \operatorname{supp}(\rho)$, $\langle x_i | y_j \rangle = 0$. And equivalently, $\operatorname{Ker}(\sigma) \perp \operatorname{supp}(\rho)$.
\end{lemma}

\begin{IEEEproof}
   First we assume that $\operatorname{supp}(\rho) \subseteq \operatorname{supp}(\sigma)$. And this is equivalent to $\operatorname{Ker}(\sigma) \subseteq \operatorname{ker}(\rho)$. 

   Then for any $|y_j \rangle \in \operatorname{Ker}(\sigma)$ and every $|x_i \rangle \in \operatorname{supp}(\rho)$,
   \begin{equation*}
       |x_i\rangle \in \operatorname{Ker}(\rho)^{\perp} \implies \langle x_i | y_j \rangle = 0.
   \end{equation*}
   Conversely, we assume that for all $|y_j\rangle \in \operatorname{Ker}(\sigma)$ and every $|x_i\rangle \in \operatorname{supp}(\rho)$, $\langle x_i | y_j \rangle = 0$.

   Then $\operatorname{Ker}(\sigma) \subseteq \operatorname{supp}(\rho)^{\perp} = \operatorname{Ker}(\rho)$. Considering the orthogonal complements of both the sets, we have $\operatorname{supp}(\rho) \subseteq \operatorname{supp}(\sigma)$.
\end{IEEEproof}

\subsection{Properties of Quantum Relative-$\alpha$-Entropy}
In this subsection, we study several fundamental mathematical properties of the quantum relative $\alpha$-entropy and explore its connections with quantities that are central to quantum information theory. We begin by examining the non-negativity of $S_{\alpha}(\rho \| \sigma)$.
For comparison, recall that the non-negativity of Umegaki’s relative entropy $U(\rho \| \sigma)$
is a direct consequence of Klein’s inequality.

\begin{lemma}
\label{Lemma:positivity}
The quantum relative $\alpha$-entropy is non-negative, that is, for any two quantum states $\rho$ and $\sigma$,
\[
S_{\alpha}(\rho \| \sigma) \geq 0.
\]
The equality holds in the above if and only if $\rho = \sigma$.
\end{lemma}

\begin{IEEEproof}
To prove the non-negativity of $S_\alpha(\rho \| \sigma)$, it suffices to show
\begin{equation}
\operatorname{Tr}(\rho \sigma^{\alpha-1})
\ge
\{\operatorname{Tr}(\rho^\alpha)\}^{1/\alpha}
\{\operatorname{Tr}(\sigma^\alpha)\}^{(\alpha-1)/\alpha},
\label{keyineq}
\end{equation}
for all $\alpha>0$, $\alpha\neq1$.

\medskip
\noindent
\textit{Step 1: Commuting case.}
Assume first that $\rho$ and $\sigma$ commute. Then they admit a joint spectral decomposition and
\[
\operatorname{Tr}(\rho^\alpha)
=
\operatorname{Tr}\!\left[(\rho \sigma^{\alpha-1})^\alpha (\sigma^\alpha)^{1-\alpha}\right].
\]
Applying Hölder's inequality for traces yields
\[
\operatorname{Tr}(\rho^\alpha)
\le
\big[\operatorname{Tr}(\rho \sigma^{\alpha-1})\big]^\alpha
\big[\operatorname{Tr}(\sigma^\alpha)\big]^{1-\alpha}
\quad \text{for } 0<\alpha<1,
\]
with the reverse inequality for $\alpha>1$. Rearranging gives
\eqref{keyineq} for all $\alpha>0$, $\alpha\neq1$.

\medskip
\noindent
\textit{Step 2: General case.}
Let
\[
\rho=\sum_i p_i |x_i\rangle\langle x_i|,
\qquad
\sigma=\sum_j q_j |y_j\rangle\langle y_j|.
\]
Then
\[
\operatorname{Tr}(\rho \sigma^{\alpha-1})
=
\sum_{i,j} p_i q_j^{\alpha-1} |\langle x_i|y_j\rangle|^2.
\]
Define $M_{ij}:=|\langle x_i|y_j\rangle|^2$. Here $M$ is doubly stochastic, that is,
$\sum_j M_{ij}=1$ and $\sum_i M_{ij}=1$.

Since $t\mapsto t^{\alpha-1}$ is convex for $\alpha>1$ and concave for $0<\alpha<1$, mixing by a doubly stochastic matrix yields
\[
\sum_{i,j} p_i q_j^{\alpha-1} M_{ij}
\begin{cases}
\le \sum_{i,j} p_i q_j^{\alpha-1}, & \alpha>1,\\[4pt]
\ge \sum_{i,j} p_i q_j^{\alpha-1}, & 0<\alpha<1.
\end{cases}
\]
Multiplying by $\frac{\alpha}{1-\alpha}$ reverses the inequality in the first case and preserves it in the second, so that in both regimes
\[
\frac{\alpha}{1-\alpha}
\log \sum_{i,j} p_i q_j^{\alpha-1} M_{ij}
\ge
\frac{\alpha}{1-\alpha}
\log \sum_{i,j} p_i q_j^{\alpha-1}.
\]
Combining this with the commuting-case bound establishes \eqref{keyineq} for arbitrary density matrices. 
\end{IEEEproof}

\begin{lemma}
\label{Lemma:Properties}
The quantum relative $\alpha$-entropy satisfies the following properties.
\begin{enumerate}
\item \emph{Additivity under tensor products:}
For density operators $\rho, \sigma, \tau, \omega$ satisfying
$\operatorname{supp}(\rho) \subseteq \operatorname{supp}(\sigma)$ and
$\operatorname{supp}(\tau) \subseteq \operatorname{supp}(\omega)$,
\[
S_{\alpha}(\rho \otimes \tau \| \sigma \otimes \omega)
= S_{\alpha}(\rho \| \sigma) + S_{\alpha}(\tau \| \omega).
\]

\item \emph{Unitary invariance:}
For any unitary operator $U$ on $\mathcal{H}$,
\[
S_{\alpha}(U \rho U^{*} \| U \sigma U^{*})
= S_{\alpha}(\rho \| \sigma).
\]
\end{enumerate}
\end{lemma}

\begin{IEEEproof}
Following the definition in \eqref{quantumjonesetal}, we have 
\begin{align*}
    & S_{\alpha}(\rho \otimes \tau|| \sigma \otimes \omega) \nonumber \\
    &= \frac{\alpha}{1 - \alpha}\log [\operatorname{Tr} \{(\rho \otimes \tau) (\sigma \otimes \omega)^{\alpha - 1}\}] - \frac{1}{1 - \alpha} \log [\operatorname{Tr} \{(\rho \otimes \tau) ^\alpha\}] + \log [\operatorname{Tr} \{(\sigma \otimes \omega)^\alpha\}] \nonumber \\
    & \stackrel{(a)}{=} \frac{\alpha}{1 - \alpha}\log [\operatorname{Tr} \{(\rho \otimes \tau) (\sigma^{\alpha - 1} \otimes \omega^{\alpha - 1})\}] - \frac{1}{1 - \alpha} \log [\operatorname{Tr} (\rho^\alpha \otimes \tau^\alpha)] + \log [\operatorname{Tr} (\sigma^\alpha \otimes \omega^\alpha)]  \nonumber \\ 
    & \stackrel{(b)}{=} \frac{\alpha}{1 - \alpha}\log [\operatorname{Tr} \{(\rho \sigma^{\alpha - 1}) \otimes (\tau \omega^{\alpha - 1})\}] - \frac{1}{1 - \alpha} \log [\operatorname{Tr} (\rho^\alpha \otimes \tau^\alpha)] +  \log [\operatorname{Tr} (\sigma^\alpha \otimes \omega^\alpha)] \nonumber \\ 
    & \stackrel{(c)}{=} \frac{\alpha}{1 - \alpha}\log [ \operatorname{Tr} (\rho \sigma^{\alpha - 1}) \operatorname{Tr} (\tau \omega^{\alpha - 1})] - \frac{1}{1 - \alpha} \log [\operatorname{Tr} (\rho^\alpha) \operatorname{Tr} (\tau^\alpha)] + \log [\operatorname{Tr} (\sigma^\alpha) \operatorname{Tr} (\omega^\alpha)] \nonumber \\ 
    & = \frac{\alpha}{1 - \alpha}\log [ \operatorname{Tr} (\rho \sigma^{\alpha - 1}) ] - \frac{1}{1 - \alpha} \log [\operatorname{Tr} (\rho^\alpha) ] + \log [\operatorname{Tr} (\sigma^\alpha) ]  + \frac{\alpha}{1 - \alpha}\log [\operatorname{Tr} (\tau \omega^{\alpha - 1})] \nonumber \\
    &~~~ - \frac{1}{1 - \alpha} \log [\operatorname{Tr} (\tau^\alpha)] + \log [\operatorname{Tr} (\omega^\alpha)] \nonumber \\
    & = S_{\alpha}(\rho || \sigma) + S_{\alpha}(\tau || \omega). \nonumber
\end{align*}
The equality in $(a)$ follows from the identity
$(A \otimes B)^{m} = A^{m} \otimes B^{m}$,
which holds for any two complex positive semi-definite matrices $A$ and $B$ and
any real $m$. For negative values of $m$, this identity is well-defined when
restricted to the support of the operators. The step $(b)$ is justified by the property
$(A \otimes B)(C \otimes D) = (AC) \otimes (BD)$,
whenever the products $AC$ and $BD$ are well-defined. Finally, the equality in
$(c)$ follows directly from the trace factorization rule
$\operatorname{Tr}(A \otimes B) = \operatorname{Tr}(A)\operatorname{Tr}(B)$.
This completes the proof of the first statement of the lemma. 

The second statement follows analogously from the representation in \eqref{quantumjonesetal} along with the properties that $(U \rho U^{*})^m = U \rho^m U^{*}$, and that $\operatorname{Tr}(U \rho U^{*}) = \operatorname{Tr}(\rho)$ for any unitary operator $U$ on $\mathcal{H}$.
\end{IEEEproof}

It should be noted that, while all the divergences from the quantum $f$-divergence class \eqref{standard_f} are unitarily invariant, they are not always additive under tensor products. Some examples include the quantum $\chi^2$-divergences \cite{Temme10}, quantum Hellinger divergences \cite{Osaka25}. 

\begin{remark}
    The quantum relative $\alpha$-entropy, $S_{\alpha}(\rho || \sigma)$ is not generally monotonic in $\alpha,$ as demonstrated in Figure \ref{fig1}. Table \ref{Tab:Sa_behavior} exhibits the different behavior of $S_\alpha(\rho||\sigma)$ over sets of increasing values of $\alpha$. 
\end{remark}

\begin{table}[ht]
\centering
\renewcommand{\arraystretch}{1.5} 
\setlength{\tabcolsep}{10pt}       
\begin{tabular}{|c|c|c|c|}
\hline
Density Matrices & $\alpha$-values & $S_\alpha$-values  & Behavior of $S_\alpha$ \\
\hline
\multirow{3}{*}{
$
\begin{aligned}
\rho &= 
\begin{pmatrix}
0 & 0 \\ 0 & 1
\end{pmatrix}
\qquad
\sigma =
\begin{pmatrix}
3/4 & 0 \\ 0 & 1/4
\end{pmatrix}
\end{aligned}
$
}
& $0.7$ & $1.660$  & \multirow{3}{*}{Increasing} \\ \cline{2-3}
& $0.9$ & $1.886$ &  \\ \cline{2-3}
& $1.2$ & $2.243$ &  \\
\hline
\multirow{3}{*}{
$
\begin{aligned}
\rho &= 
\begin{pmatrix}
1 & 0 \\ 0 & 0
\end{pmatrix}
\qquad
\sigma =
\begin{pmatrix}
3/4 & 0 \\ 0 & 1/4
\end{pmatrix}
\end{aligned}
$
}
& $0.5$ & $0.6572$ & \multirow{3}{*}{Decreasing} \\ \cline{2-3}
& $0.7$ & $0.5495$ &  \\ \cline{2-3}
& $0.9$ & $0.4531$ &  \\
\hline

\multirow{3}{*}{
$
\begin{aligned}
\rho = \begin{pmatrix}
    0.8 & 0.2 \\ 0.2 & 0.2
\end{pmatrix}
\qquad
\sigma = \begin{pmatrix}
    0.6 & 0 \\ 0 & 0.4
\end{pmatrix}
\end{aligned}
$
}
& $1.5$ & $0.3311$ & \multirow{3}{*}{Oscillating} \\ \cline{2-3}
& $2$ & $0.3334$ &  \\ \cline{2-3}
& $3$ & $0.3076$ &  \\
\hline
\end{tabular}
\vspace{0.3 cm}
\caption{Behavior of $S_\alpha(\rho\|\sigma)$ for different density matrix pairs}
\label{Tab:Sa_behavior}
\end{table}

\begin{figure}
\centering
\includegraphics[width=0.9\textwidth]{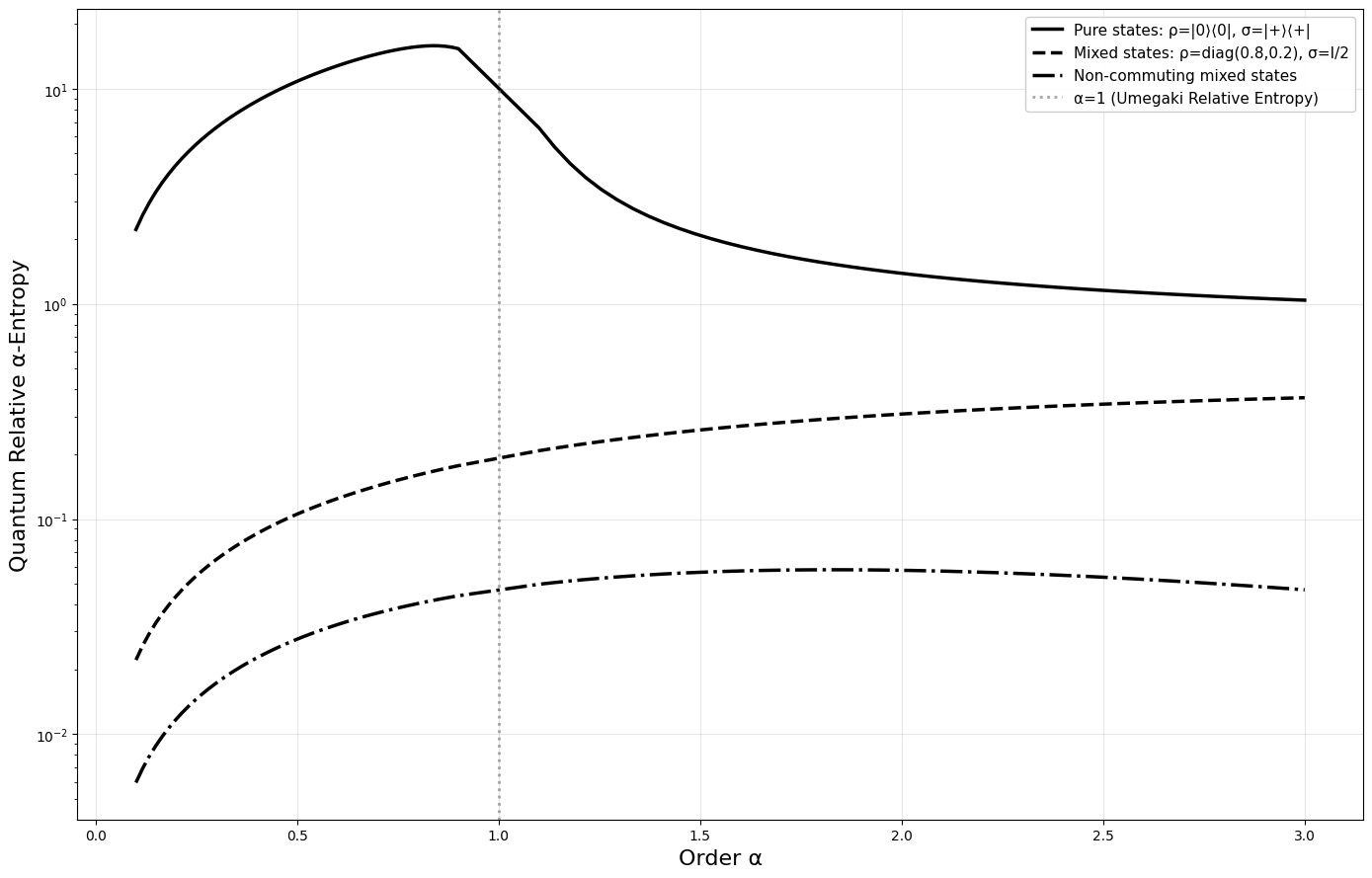}
\caption{The Quantum Relative $\alpha$-Entropy as a function of its order for three different sets of quantum states.}
\label{fig1}
\end{figure}

\begin{lemma}
\label{Lemma:Constants}
    For any two positive constants $k_1$ and $k_2,$ $S_{\alpha}(k_1 \rho || k_2 \sigma) = S_{\alpha}(\rho || \sigma).$
\end{lemma}

\begin{IEEEproof}
Using the definition of the quantum relative $\alpha$-entropy \eqref{quantumjonesetal}, we have 
\begin{eqnarray}
 {S_{\alpha}(k_1 \rho || k_2 \sigma)} 
   & = & \frac{\alpha}{1 - \alpha} \log \operatorname{Tr} (k_1 \rho k_2^{\alpha - 1} \sigma^{\alpha - 1}) - \frac{1}{1 - \alpha} \log \operatorname{Tr} (k_1^{\alpha} \rho^\alpha) + \log \operatorname{Tr} (k_2^{\alpha}\sigma^\alpha)  \nonumber \\
   & = & \frac{\alpha}{1 - \alpha} \log k_1 + \frac{\alpha}{1 - \alpha} \log \operatorname{Tr}(\rho \sigma^{\alpha - 1}) + \frac{\alpha}{1 - \alpha} \log k_2^{\alpha - 1} - \frac{1}{1 - \alpha} \log \operatorname{Tr}(\rho^\alpha) \nonumber \\
   & - & \frac{1}{1-\alpha} \log k_1^\alpha + \log k_2^\alpha + \log \operatorname{Tr}(\sigma^\alpha) \nonumber \\
   & = & S_{\alpha}(\rho || \sigma) + \frac{\alpha}{1 - \alpha} \log k_1 - \frac{\alpha}{1 - \alpha} \log k_1 + \alpha \log k_2 + \frac{\alpha(\alpha-1)}{1-\alpha} \log k_2 \nonumber \\
   & = & S_{\alpha}(\rho || \sigma). \nonumber
\end{eqnarray}
\end{IEEEproof}

\begin{remark}
    \begin{itemize}
        \item[1.] The lemma above implies that the quantum relative $\alpha$-entropy depends only on the relative geometry, overlapping of the density matrices $\rho$ and $\sigma$, not on their overall magnitudes. 
        \item[2.] This property does not hold for the members of $f$-divergence class. For example, the Petz-R\'enyi-$\alpha$-relative entropy $\Hat{D}_{\alpha}(\rho||\sigma)$ is affine under scaling, but not invariant.
    \end{itemize}
\end{remark}


\subsection{A Nonlinear Convexity Framework for Quantum Divergences} 

A real-valued function $f : D \to \mathbb{R}$, where 
$D \subseteq \mathbb{R}$,
is said to be \emph{convex} if for all $x,y \in D$ and $t \in
[0,1]$,
\[
f(t x + (1-t)y) \leq t f(x) + (1-t) f(y).
\]
The set $D \subseteq \mathbb{R}$ itself is called \emph{convex} if
$t x + (1-t)y \in D$ for all $x,y \in D$ and $t \in [0,1]$.
Analogously, a function $f(x,y)$ is said to be \emph{jointly convex} on
$D \times D$ if, for all $x_1,x_2,y_1,y_2 \in D$ and $t \in [0,1]$,
\[
f\bigl(t x_1 + (1-t)x_2, t y_1 + (1-t)y_2\bigr)
\leq
t f(x_1,y_1) + (1-t)f(x_2,y_2).
\]

It is easy to check that the set of all density operators forms a convex set.
Moreover, the joint convexity of several quantum divergences, such as
Umegaki’s relative entropy \eqref{Umegaki}, the Petz-R\'enyi
$\alpha$-relative entropy \eqref{RenyiRelativeentropy}, and the sandwiched
R\'enyi divergence \eqref{Sandwiched}, has been extensively studied in the
literature; see, for example, \cite{Uhlmann77,MullerMartin2013}.

In contrast, the quantum relative $\alpha$-entropy
$S_{\alpha}(\rho \| \sigma)$ defined in \eqref{quantumjonesetal}
is neither convex in $\rho$ nor in $\sigma$ for general values of
$\alpha > 0$, $\alpha \neq 1$.
In the special case where $\sigma$ is fixed,
$S_{\alpha}(\rho \| \sigma)$ is convex in $\rho$ for $\alpha \in (0,1)$.
However, $S_{\alpha}(\rho \| \sigma)$ fails to be jointly convex, primarily due to its multiplicative, rather than linear dependence on the arguments $\rho$ and $\sigma$.
The additive mixing required for standard joint convexity disrupts the algebraic structure underlying $S_{\alpha}(\rho \| \sigma)$.

Motivated by this observation, we introduce a modified notion of convexity that is compatible with the multiplicative structure of the divergence $S_\alpha$. Specifically, we replace linear convex combinations by normalized products
of density operators raised to non-linear powers as described in the following definition.

\begin{definition}
\label{Defn:ConvexSet}
A set $\mathcal{A}$ of density operators is said to be \emph{generalized
convex} if, for any $\rho,\sigma \in \mathcal{A}$ and any $t \in [0,1]$, the
operator
\[
M_{\rho,\sigma}^{t}
:= \frac{\rho^{t}\sigma^{1-t}}
{\operatorname{Tr}(\rho^{t}\sigma^{1-t})}
\]
also belongs to $\mathcal{A}$.
\end{definition}

\begin{remark}
\begin{itemize}
    \item [1.] The generalized convex combination defined
above, based on probability densities, is well known in the 
context of non-extensive statistical physics. See, for
example, \cite{Naudts04, KumarS15J2}.
\item[2.]  For any two arbitrary density operators $\rho$ and $\sigma$, the matrix
$M_{\rho,\sigma}^{t}$ is not necessarily a valid density operator. Although $\rho^{t}$ and $\sigma^{1-t}$ are positive semidefinite for $\rho,\sigma \geq 0$, their product need not be Hermitian or positive semidefinite unless $\rho$ and $\sigma$ commute.
Hence, $M_{\rho,\sigma}^{t}$ defines a density operator if and only if
$\rho$ and $\sigma$ commute.
The normalization factor
$\operatorname{Tr}(\rho^{t}\sigma^{1-t})$ ensures that
$\operatorname{Tr}(M_{\rho,\sigma}^{t}) = 1$ whenever the product is
well-defined.
\end{itemize}
\end{remark}

\begin{lemma}
Any generalized convex set $\mathcal{A}$ is a proper subset of the set of all
density operators and consists solely of mutually commuting density operators. Consequently, all elements of $\mathcal{A}$ are simultaneously diagonalizable by a common unitary transformation.
\end{lemma}

We now restrict attention to the quantum relative $\alpha$-entropy $S_{\alpha}(\rho \| \sigma)$ defined on a generalized convex set $\mathcal{A}$ and introduce a corresponding generalized notion of joint convexity adapted to this setting.

\begin{lemma}
    Let $\rho, \sigma, \tau, \omega$ be density matrices in $\mathcal{A}$ and $t \in [0,1].$ Then for $\alpha < 1,$
\begin{eqnarray}
\label{Convexity}
    {S_{\alpha}(M_{\rho, \sigma}^t || M_{\tau, \omega}^t)}
    & \leq & t S_{\alpha}(\rho||\tau) + (1-t) S_{\alpha}(\sigma||\omega) + \frac{1}{\alpha-1} \log (Z_{\rho, \sigma}^t) + \log (Z_{\tau, \omega}^t),
\end{eqnarray}
where $Z_{\rho, \sigma}^t$ is a real number number defined as:
\begin{equation*}
    Z_{\rho, \sigma}^t : = \operatorname{Tr} \Bigg[ \left\{\Big(\frac{\rho}{||\rho||_\alpha}\Big)^t \Big(\frac{\sigma}{||\sigma||_\alpha}\Big)^{1-t} \right\}^\alpha\Bigg].
\end{equation*}
The inequality is reversed for $\alpha > 1.$
\end{lemma}

\begin{IEEEproof}
Here 
\begin{eqnarray}
    ||M_{\rho, \sigma}^t||_{\alpha} = \Bigg[ \operatorname{Tr} \Big( \frac{\rho^t \sigma^{1-t}}{\operatorname{Tr} (\rho^t \sigma^{1-t})} \Big)^{\alpha} \Bigg]^{1/\alpha} = \frac{[\operatorname{Tr} \{(\rho^t \sigma^{1-t})^{\alpha}\}]^{1/\alpha}} {\operatorname{Tr} (\rho^t \sigma^{1-t})}. \nonumber
\end{eqnarray}

and
\begin{equation*}
    \frac{M_{\rho, \sigma}^t}{||M_{\rho, \sigma}^t||_{\alpha}} : = \frac{\rho^t \sigma^{1-t}}{[\operatorname{Tr} (\rho^t \sigma^{1-t})]^{1/\alpha}}.
\end{equation*}

Using the definition of the quantum relative $\alpha$-entropy, from \eqref{1:quantumjones1} we have 
\begin{eqnarray}
\lefteqn {S_{\alpha}(M_{\rho, \sigma}^t || M_{\tau, \omega}^t)} \nonumber\\
   & = & \frac{\alpha}{1 - \alpha} \log \operatorname{Tr} \Bigg[  \Big(\frac{M_{\rho, \sigma}^t}{||M_{\rho, \sigma}^t||_{\alpha}}\Big) \Big(\frac{M_{\tau, \omega}^t}{||M_{\tau, \omega}^t||_{\alpha}}\Big)^{\alpha-1} \Bigg]  \nonumber \\
   & = & \frac{\alpha}{1 - \alpha} \log \operatorname{Tr} \Bigg[  \Big(\frac{\rho^t \sigma^{1-t}}{[\operatorname{Tr} (\rho^t \sigma^{1-t})]^{1/\alpha}}\Big) \Big(\frac{\tau^t \omega^{1-t}}{[\operatorname{Tr} (\tau^t \omega^{1-t})]^{1/\alpha}}\Big)^{\alpha-1} \Bigg]  \nonumber \\
   & = & \frac{\alpha}{1 - \alpha} \log \operatorname{Tr} \Bigg[ \left\{ \Big(\frac{\rho}{||\rho||_\alpha}\Big)^t \Big(\frac{\sigma}{||\sigma||_\alpha}\Big)^{1-t} \right\} \left\{ \Big(\frac{\tau}{||\tau||_\alpha}\Big)^t \Big(\frac{\omega}{||\omega||_\alpha}\Big)^{1-t} \right\}^{\alpha-1} \Bigg] \nonumber \\
   & + & \frac{\alpha}{1 - \alpha} \log \Bigg[ \Big(\operatorname{Tr} \{(\rho^t \sigma^{1-t})^{\alpha}\}\Big)^{-1/\alpha}  \Big(\operatorname{Tr} \{(\tau^t \omega^{1-t})^{\alpha}\}\Big)^{1-\alpha/\alpha}  \Bigg] \nonumber \\
   & + & \frac{\alpha}{1 - \alpha} \log \Bigg[ \Big(\frac{1}{||\rho||_\alpha}\Big)^{-t} \Big(\frac{1}{||\sigma||_\alpha}\Big)^{t-1} \Big(\frac{1}{||\tau||_\alpha}\Big)^{t(1-\alpha)} \Big(\frac{1}{||\omega||_\alpha}\Big)^{(1-t)(1-\alpha)}\Bigg] \nonumber \\
   & = & \frac{\alpha}{1 - \alpha} \log \operatorname{Tr} \Bigg[ \left\{ \Big(\frac{\rho}{||\rho||_\alpha}\Big)^t \Big(\frac{\sigma}{||\sigma||_\alpha}\Big)^{1-t} \right\} \left\{ \Big(\frac{\tau}{||\tau||_\alpha}\Big)^t \Big(\frac{\omega}{||\omega||_\alpha}\Big)^{1-t} \right\}^{\alpha-1} \Bigg] \nonumber \\
   & + & \frac{\alpha}{1 - \alpha} \log \Bigg[\Big(\operatorname{Tr} \left\{\Big(\frac{\rho}{||\rho||_\alpha}\Big)^t \Big(\frac{\sigma}{||\sigma||_\alpha}\Big)^{1-t} \right\}^\alpha \Big)^{-1/\alpha} \Big(\operatorname{Tr} \left\{\Big(\frac{\tau}{||\tau||_\alpha}\Big)^t \Big(\frac{\omega}{||\omega||_\alpha}\Big)^{1-t} \right\}^\alpha \Big)^{1-\alpha/\alpha} \Bigg] \nonumber \\
   & = & \frac{\alpha}{1 - \alpha} \log \operatorname{Tr} \Bigg[ \left\{ \Big(\frac{\rho}{||\rho||_\alpha}\Big)^t \Big(\frac{\sigma}{||\sigma||_\alpha}\Big)^{1-t} \right\} \left\{ \Big(\frac{\tau}{||\tau||_\alpha}\Big)^t \Big(\frac{\omega}{||\omega||_\alpha}\Big)^{1-t} \right\}^{\alpha-1} \Bigg] \nonumber \\
   & + & \frac{1}{\alpha-1} \log \operatorname{Tr} \Bigg[ \left\{\Big(\frac{\rho}{||\rho||_\alpha}\Big)^t \Big(\frac{\sigma}{||\sigma||_\alpha}\Big)^{1-t} \right\}^\alpha\Bigg] + \log \operatorname{Tr} \Bigg[\left\{\Big(\frac{\tau}{||\tau||_\alpha}\Big)^t \Big(\frac{\omega}{||\omega||_\alpha}\Big)^{1-t} \right\}^\alpha \Bigg]. \nonumber
\end{eqnarray}
So the expression reduces to 
\begin{eqnarray}
 {S_{\alpha}(M_{\rho, \sigma}^t || M_{\tau, \omega}^t)} 
   & = & \frac{\alpha}{1 - \alpha} \log \operatorname{Tr} \Bigg[ \left\{ \Big(\frac{\rho}{||\rho||_\alpha}\Big)^t \Big(\frac{\sigma}{||\sigma||_\alpha}\Big)^{1-t} \right\} \left\{ \Big(\frac{\tau}{||\tau||_\alpha}\Big)^t \Big(\frac{\omega}{||\omega||_\alpha}\Big)^{1-t} \right\}^{\alpha-1} \Bigg]  \nonumber \\
   & + & \frac{1}{\alpha-1} \log (Z_{\rho, \sigma}^t) + \log (Z_{\tau, \omega}^t). \nonumber
\end{eqnarray}

For any two positive semi-definite, Hermitian matrices $A$ and $B,$ which are commutative, $(AB)^m = A^m B^m,$ for any real number $m.$ So we obtain 
\begin{align*}
    &\operatorname{Tr} \Bigg[ \left\{ \Big(\frac{\rho}{||\rho||_\alpha}\Big)^t \Big(\frac{\sigma}{||\sigma||_\alpha}\Big)^{1-t} \right\} \left\{ \Big(\frac{\tau}{||\tau||_\alpha}\Big)^t \Big(\frac{\omega}{||\omega||_\alpha}\Big)^{1-t} \right\}^{\alpha-1} \Bigg] \nonumber \\
    &= \operatorname{Tr} \Bigg[ \Big(\frac{\rho}{||\rho||_\alpha}\Big)^t \Big(\frac{\sigma}{||\sigma||_\alpha}\Big)^{1-t} \Big(\frac{\tau}{||\tau||_\alpha}\Big)^{t(\alpha-1)} \Big(\frac{\omega}{||\omega||_\alpha}\Big)^{(1-t)(\alpha-1)} \Bigg] \nonumber \\
    &= \operatorname{Tr} \Bigg[ \Big(\frac{\rho}{||\rho||_\alpha}\Big)^t \Big(\frac{\tau}{||\tau||_\alpha}\Big)^{t(\alpha-1)} \Big(\frac{\sigma}{||\sigma||_\alpha}\Big)^{1-t} \Big(\frac{\omega}{||\omega||_\alpha}\Big)^{(1-t)(\alpha-1)} \Bigg]
    \nonumber \\
    &= \operatorname{Tr} \Bigg[ \left\{ \Big(\frac{\rho}{||\rho||_\alpha}\Big) \Big(\frac{\tau}{||\tau||_\alpha}\Big)^{(\alpha-1)} \right\}^{t} \left\{\Big(\frac{\sigma}{||\sigma||_\alpha}\Big) \Big(\frac{\omega}{||\omega||_\alpha}\Big)^{(\alpha-1)}\right\}^{1-t} \Bigg] \nonumber \\
    &\leq \Bigg[ \operatorname{Tr} \left\{ \Big(\frac{\rho}{||\rho||_\alpha}\Big) \Big(\frac{\tau}{||\tau||_\alpha}\Big)^{(\alpha-1)} \right\} \Bigg]^t \Bigg[ \operatorname{Tr} \left\{\Big(\frac{\sigma}{||\sigma||_\alpha}\Big) \Big(\frac{\omega}{||\omega||_\alpha}\Big)^{(\alpha-1)}\right\} \Bigg]^{1-t}, \nonumber
\end{align*}
where the last inequality is deduced following the fact that
\[
\operatorname{Tr}[(A^m) (B^{1-m})] \leq [\operatorname{Tr}(A)]^m [\operatorname{Tr}(B)]^{1-m}
\]
for any positive semi-definite matrices $A$ and $B,$ with $m \in [0,1]$ \cite{zhang11}, which is again backed by H\"{o}lder's inequality.

So when $\alpha < 1,$
\begin{align*}
    &\frac{\alpha}{1 - \alpha} \log \operatorname{Tr} \Bigg[ \left\{ \Big(\frac{\rho}{||\rho||_\alpha}\Big)^t \Big(\frac{\sigma}{||\sigma||_\alpha}\Big)^{1-t} \right\} \left\{ \Big(\frac{\tau}{||\tau||_\alpha}\Big)^t \Big(\frac{\omega}{||\omega||_\alpha}\Big)^{1-t} \right\}^{\alpha-1} \Bigg] \nonumber \\
    &\leq \frac{\alpha}{1 - \alpha} \log \Bigg[ \operatorname{Tr} \left\{ \Big(\frac{\rho}{||\rho||_\alpha}\Big) \Big(\frac{\tau}{||\tau||_\alpha}\Big)^{(\alpha-1)} \right\} \Bigg]^t \Bigg[ \operatorname{Tr} \left\{\Big(\frac{\sigma}{||\sigma||_\alpha}\Big) \Big(\frac{\omega}{||\omega||_\alpha}\Big)^{(\alpha-1)}\right\} \Bigg]^{1-t} \nonumber \\
    &= t \Big(\frac{\alpha}{1 - \alpha}\Big) \log \Bigg[ \operatorname{Tr} \left\{ \Big(\frac{\rho}{||\rho||_\alpha}\Big) \Big(\frac{\tau}{||\tau||_\alpha}\Big)^{(\alpha-1)} \right\} \Bigg] + (1-t) \Big(\frac{\alpha}{1 - \alpha}\Big)
    \log \Bigg[ \operatorname{Tr} \left\{ \Big(\frac{\sigma}{||\sigma||_\alpha}\Big) \Big(\frac{\omega}{||\omega||_\alpha}\Big)^{(\alpha-1)} \right\} \Bigg] \nonumber \\
    &= t S_{\alpha}(\rho||\tau) + (1-t) S_{\alpha}(\sigma||\omega).
\end{align*}

And finally, we have 
\begin{eqnarray}
    {S_{\alpha}(M_{\rho, \sigma}^t || M_{\tau, \omega}^t)}
    & \leq & t S_{\alpha}(\rho||\tau) + (1-t) S_{\alpha}(\sigma||\omega) + \frac{1}{\alpha-1} \log (Z_{\rho, \sigma}^t) + \log (Z_{\tau, \omega}^t). \nonumber
\end{eqnarray}

For $\alpha > 1,$ the fraction $\frac{\alpha}{1 - \alpha} < 0.$ This flips the direction of the argument, reversing the final inequality. 
\end{IEEEproof}

\begin{remark}
    \begin{eqnarray}
    Z_{\rho, \sigma}^t = \operatorname{Tr} \Bigg[ \left\{\Big(\frac{\rho}{||\rho||_\alpha}\Big)^t \Big(\frac{\sigma}{||\sigma||_\alpha}\Big)^{1-t} \right\}^\alpha\Bigg] \nonumber
    & = & \operatorname{Tr} \Bigg[ \Big(\frac{\rho}{||\rho||_\alpha}\Big)^{\alpha t} \Big(\frac{\sigma}{||\sigma||_\alpha}\Big)^{\alpha(1-t)} \Bigg] \\ \nonumber
    & = & \operatorname{Tr} \Big[(\rho^{\alpha})^t (\sigma^\alpha)^{1-t}
    \Big] \Big[\frac{1}{||\rho||_\alpha}\Big]^{\alpha t} \Big[\frac{1}{||\sigma||_\alpha}\Big]^{\alpha(1-t)} \\ \nonumber
    & \leq & \Big[\operatorname{Tr} (\rho^\alpha)\Big]^t \Big[\operatorname{Tr} (\sigma^\alpha)\Big]^{1-t}  \Big[\frac{1}{\operatorname{Tr} (\rho^\alpha)} \Big]^t \Big[\frac{1}{\operatorname{Tr} (\sigma^\alpha)} \Big]^{1-t} \\ \nonumber
    & = & 1.
\end{eqnarray}
Thus, we have
\begin{equation*}
    \log (Z_{\rho, \sigma}^t) \leq 0.
\end{equation*}
\end{remark}

When the density matrices are all classical states, the expression \eqref{Convexity} represents this notion of generalized convexity of the classical relative-$\alpha$-entropy, $J_{\alpha}(p || q)$ as in \eqref{Jones}. The Petz-R\'enyi-$\alpha$ relative entropy $\hat{D}_\alpha(\rho||\sigma)$ and Sandwiched R\'enyi divergence, $D_{\alpha}^{*}(\rho||\sigma)$ also follow this notion of generalized convexity, as stated below.

\begin{corollary}
\label{RenyiGenConvex}
 Let $\rho, \sigma, \tau, \omega$ be density matrices in $\mathcal{A}$ and $t \in [0,1].$ Then for $\alpha > 1,$
\begin{eqnarray}
    {\hat{D}_\alpha(M_{\rho, \sigma}^t || M_{\tau, \omega}^t)}
    & \leq & t \hat{D}_\alpha(\rho||\tau) + (1-t) \hat{D}_\alpha(\sigma||\omega) + \frac{\alpha}{1-\alpha} \log [\operatorname{Tr}(\rho^t \sigma^{(1-t)})] + \log [\operatorname{Tr}(\tau^t \omega^{(1-t)})], \nonumber
\end{eqnarray} 
where $\hat{D}_\alpha(\rho||\sigma)$ is the Petz-R\'enyi-$\alpha$ relative entropy \eqref{RenyiRelativeentropy}. The inequality is reversed when $\alpha < 1$.
\end{corollary}

\begin{IEEEproof}
    Using the definition \eqref{RenyiRelativeentropy}, we have 

\begin{equation}
    {\hat{D}_\alpha(M_{\rho, \sigma}^t || M_{\tau, \omega}^t)} = \frac{1}{\alpha-1} \log \operatorname{Tr}\Bigg[\left\{\frac{\rho^t \sigma^{(1-t)}}{\operatorname{Tr}(\rho^t \sigma^{(1-t)})}\right\}^\alpha \left\{\frac{\tau^t \omega^{(1-t)}}{\operatorname{Tr}(\tau^t \omega^{(1-t)})}\right\}^{1-\alpha} \Bigg]. \nonumber
\end{equation}

Here 
\begin{align*}
    &\operatorname{Tr}\Bigg[\left\{\frac{\rho^t \sigma^{(1-t)}}{\operatorname{Tr}(\rho^t \sigma^{(1-t)})}\right\}^\alpha \left\{\frac{\tau^t \omega^{(1-t)}}{\operatorname{Tr}(\tau^t \omega^{(1-t)})}\right\}^{1-\alpha} \Bigg] \nonumber \\
    &= \Big[\frac{1}{\operatorname{Tr}(\rho^t \sigma^{(1-t)})}\Big]^\alpha \Big[\frac{1}{\operatorname{Tr}(\tau^t \omega^{(1-t)})}\Big]^{1-\alpha} \operatorname{Tr}\Big[(\rho^{\alpha t} \sigma^{\alpha(1-t)})(\tau^{t(1-\alpha)} \omega^{(1-t)(1-\alpha)})\Big]  \nonumber \\
    &= \Big[\operatorname{Tr}(\rho^t \sigma^{(1-t)})\Big]^{-\alpha}\Big[\operatorname{Tr}(\tau^t \omega^{(1-t)})\Big]^{\alpha-1} \operatorname{Tr}\Big[\rho^{\alpha t} \tau^{t(1-\alpha)} \sigma^{\alpha(1-t)} \omega^{(1-t)(1-\alpha)}\Big]. \nonumber
\end{align*}

And by H\"{o}lder's inequality,
\begin{align*}
     \operatorname{Tr}\Big[\rho^{\alpha t} \tau^{t(1-\alpha)} \sigma^{\alpha(1-t)} \omega^{(1-t)(1-\alpha)}\Big] 
     & = \operatorname{Tr}\Big[\left\{\rho^{\alpha} \tau^{(1-\alpha)}\right\}^t \left\{\sigma^{\alpha} \omega^{(1-\alpha)}\right\}^{1-t} \Big]\nonumber \\
     & \leq \Big[\operatorname{Tr}(\rho^{\alpha} \tau^{(1-\alpha)})\Big]^t \Big[\operatorname{Tr}(\sigma^{\alpha} \omega^{(1-\alpha)})\Big]^{1-t}.  \nonumber
\end{align*}

So when $\alpha > 1$
\begin{align*}
    & \frac{1}{\alpha-1} \log \operatorname{Tr}\Bigg[\left\{\frac{\rho^t \sigma^{(1-t)}}{\operatorname{Tr}(\rho^t \sigma^{(1-t)})}\right\}^\alpha \left\{\frac{\tau^t \omega^{(1-t)}}{\operatorname{Tr}(\tau^t \omega^{(1-t)})}\right\}^{1-\alpha} \Bigg] \nonumber \\
    & \leq \frac{1}{\alpha-1} \log \Big[\left\{\operatorname{Tr}(\rho^t \sigma^{(1-t)})\right\}^{-\alpha} \left\{\operatorname{Tr}(\tau^t \omega^{(1-t)})\right\}^{\alpha-1} \left\{\operatorname{Tr}(\rho^{\alpha} \tau^{(1-\alpha)})\right\}^t \left\{\operatorname{Tr}(\sigma^{\alpha} \omega^{(1-\alpha)})\right\}^{1-t} \Big] \nonumber \\
    & = t \Big(\frac{1}{\alpha-1}\Big) \log \Big[\operatorname{Tr}(\rho^{\alpha} \tau^{(1-\alpha)}) \Big] + (1-t) \Big(\frac{1}{\alpha-1}\Big) \log \Big[ \operatorname{Tr}(\sigma^{\alpha} \omega^{(1-\alpha)}) \Big] \nonumber\\ 
    & + \Big(\frac{\alpha}{1-\alpha}\Big) \log \Big[\operatorname{Tr}(\rho^{\alpha} \sigma^{(1-\alpha)}) \Big] + \log \Big[\operatorname{Tr}(\tau^{\alpha} \omega^{(1-\alpha)})\Big], \nonumber
\end{align*}

which proves the statement. For $\alpha < 1,$ the fraction $\frac{1}{\alpha-1} < 0,$ and the inequality is reversed.
\end{IEEEproof}

\begin{remark}
    \begin{itemize}
        \item[1.] The Sandwiched R\'enyi divergence $D_{\alpha}^{*}(\rho||\sigma)$ reduces to the Petz-R\'enyi-$\alpha$-relative entropy $\Hat{D}_{\alpha}(\rho||\sigma)$ for commutating density matrices. Consequently, when restricted to the generalized convex set $\mathcal{A}$, defined in definition \eqref{Defn:ConvexSet}, $D_{\alpha}^{*}(\rho||\sigma)$ also satisfies the inequality established above for the same range of $\alpha$. 
        \item [2.]  The Petz-R\'enyi-$\alpha$-relative entropy $\Hat{D}_{\alpha}(\rho||\sigma)$ is jointly convex in the standard sense but only for $\alpha \in [0,1]$. Corollary \ref{RenyiGenConvex} introduces an alternative, generalized joint convexity structure when $\alpha > 1$.
    \end{itemize}  
\end{remark}


\section{Quantum Relative-$\alpha$-entropy and Other Information Measures}
\label{5sec:Connectionsdiscussions}

In this section, we investigate the limiting behavior of the quantum relative $\alpha$-entropy $S_\alpha$ and its connections with other popular quantum information
measures. In particular, we analyze the limits of $S_\alpha$ as the parameter $\alpha$ approaches specific values at which well-known entropic quantities are recovered, including Umegaki's relative entropy and several R\'enyi-type quantum
divergences. These results establish the continuity properties of $S_\alpha$ with respect to $\alpha$ and position it within the broader landscape of quantum information measures.

The limiting relations derived here serve not only as consistency checks for the proposed divergence but also provide insight into its operational and interpretational significance. In particular, several R\'enyi-type quantum
divergences have been introduced in the literature. Analogous to the classical setting, we explicitly connect $S_\alpha$ to the Petz-R\'enyi-$\alpha$ relative entropy $\hat{D}_\alpha$, thereby bridging inside the class of generalized divergence measures.
Such connections may facilitate comparisons and enable potential applications of $S_\alpha$ across different inferential and information-theoretic settings.

\begin{lemma}
\label{Lemma:limitting}
For any two density operators $\rho$ and $\sigma$,
     $S_{\alpha}(\rho || \sigma) \to U(\rho||\sigma)$ as $\alpha \to 1$.
\end{lemma}

\begin{IEEEproof}
To prove the statement above we use the expression \eqref{quantumjones2} of the quantum relative $\alpha$-entropy. It is observed that 
\begin{eqnarray}
      \lim_{\alpha \to 1}  \frac{1}{1 - \alpha}\log \sum_{i=1}^n p_{i}^{\alpha}  =  - \lim_{\alpha \to 1} \frac{\sum_{i=1}^n p_{i}^{\alpha} \log p_{i}}{\sum_{i=1}^n p_{i}^{\alpha}} 
       =  - \sum_{i=1}^n p_{i} \log p_{i}, \nonumber
\end{eqnarray}

and $ \lim\limits_{\alpha \to 1} \log \sum_{j=1}^n q_{j}^{\alpha} = 0 $, since $\sum_{i =1}^n p_{i} = \sum_{j =1}^n q_{j} = 1$. 

Furthermore,
\begin{eqnarray}
      \lim_{\alpha \to 1}  \frac{\alpha}{1 - \alpha}\log \sum_{i,j =1}^n p_{i} q_{j}^{\alpha - 1} |\langle x_{i}|y_{j}\rangle|^2 
      & = & - \lim_{\alpha \to 1} \frac{\frac{\partial}{\partial \alpha}[\sum_{i,j =1}^n p_{i} q_{j}^{\alpha - 1} |\langle x_{i}|y_{j}\rangle|^2]^{\alpha}}{[\sum_{i,j =1}^n p_{i} q_{j}^{\alpha - 1} |\langle x_{i}|y_{j}\rangle|^2]^{\alpha}} \nonumber \\
      & = & - \log \sum_{i,j =1}^n p_{i} - \sum_{i,j =1}^n p_{i} \log q_{j} |\langle x_{i}|y_{j}\rangle|^2 \nonumber \\
      & = & - \sum_{i,j =1}^n p_{i} \log q_{j} |\langle x_{i}|y_{j}\rangle|^2. \nonumber
\end{eqnarray}
Accumulating all these we get, 
\begin{equation}
\label{umegakiderived}
\lim_{\alpha \to 1} S_{\alpha}(\rho || \sigma)  =  \sum_{i=1}^n p_{i} \log p_{i} - \sum_{i,j =1}^n p_{i} \log q_{j} |\langle x_{i}|y_{j}\rangle|^2.
\end{equation}

The right-hand side of \eqref{umegakiderived} complies with Umegaki's relative entropy $U(\rho||\sigma)$ as defined in \eqref{Umegaki}. Following the reasoning behind Lemma \ref{Lemma:traceequation}, we can calculate $\operatorname{Tr}(\rho \log \rho)$ and $\operatorname{Tr}(\rho \log \sigma)$ instead of $\operatorname{Tr}(\rho \sigma^{\alpha - 1})$ analogously, to show that 
\begin{equation}
     \operatorname{Tr}(\rho \log \rho) = \sum_{i=1}^n p_{i} \log p_{i} \quad \text{and} \quad  \operatorname{Tr}(\rho \log \sigma) = \sum_{i,j =1}^n p_{i} \log q_{j} |\langle x_{i}|y_{j}\rangle|^2. \nonumber
\end{equation} 
This completes the proof.
\end{IEEEproof}

For any two density operators $\rho$ and $\sigma$,  let us define the transformed matrices below.
\begin{equation}
\label{modidensity}
    \rho^{(\alpha)} = \frac{\rho^{\alpha}}{\operatorname{Tr}(\rho^{\alpha})} \quad \text{and} \quad \sigma^{(\alpha)} = \frac{\sigma^{\alpha}}{\operatorname{Tr}(\sigma^{\alpha})} \quad \text{for}~ \alpha > 0.
\end{equation}

It can easily be confirmed that $\rho^{(\alpha)}$ and $\sigma^{(\alpha)}$ are also density matrices. Further, there is a one-to-one correspondence of the above density matrices with $\rho$ and $\sigma$ respectively, since all the defining properties of $\rho$ and $\sigma$ translate to them as well. Analogous transformations based on probability distributions are popular in the context of non-extensive physics and robust statistical inference \cite{Naudts04, Tsallis88J, TsallisMP98J}. Such transformed distributions are well-known as $\alpha$-escort measure and $\alpha$-scaled measure \cite{KumarS16J, GayenK21J}.

\begin{lemma}
\label{Lemma:RenyiRelation}
    $S_{\alpha}(\rho || \sigma)$ is related to the Petz-R\'enyi-$\alpha$-relative entropy, $\Hat{D}_{\alpha}(\rho||\sigma)$ by 
    \begin{equation}
    S_{\alpha}(\rho || \sigma) = \Hat{D}_{1 / \alpha}(\rho^{(\alpha)} || \sigma^{(\alpha)}), 
    \end{equation}
    where $\Hat{D}_{\alpha}(\rho||\sigma)$ is as in \eqref{RenyiRelativeentropy} and $\rho^{(\alpha)}, \sigma^{(\alpha)}$ are as in \eqref{modidensity}.
\end{lemma}

\begin{IEEEproof}
From  \eqref{RenyiRelativeentropy}, we have 
\begin{eqnarray}
    \Hat{D}_{1 / \alpha}(\rho^{(\alpha)} || \sigma^{(\alpha)})
    & = & \frac{1}{\frac{1}{\alpha} - 1} \log \operatorname{Tr} \Big[ (\rho^{(\alpha)})^{1/\alpha} (\sigma^{(\alpha)})^{1- \frac{1}{\alpha}} \Big] \nonumber \\
    & = & \frac{\alpha}{1 - \alpha} \log \operatorname{Tr} \Big[ \frac{\rho}{(\operatorname{Tr} \rho^{\alpha})^{1/\alpha}} \Big(\frac{\sigma}{(\operatorname{Tr} \sigma^{\alpha})^{1/\alpha}} \Big)^{\alpha - 1} \Big], \nonumber
\end{eqnarray}
which coincides with \eqref{1:quantumjones1}.
\end{IEEEproof}               

\begin{lemma}
\label{Lemma:RenyiEntropyRelation}
The quantum relative $\alpha$-entropy is related to the quantum R\'enyi entropy, $R_\alpha(\rho)$, through the following relation.
    \begin{equation}
    S_{\alpha}(\rho || \sigma) = \frac{\alpha}{1 - \alpha}\log \operatorname{Tr} (\rho \sigma^{\alpha - 1}) + \log \operatorname{Tr} (\sigma ^\alpha) - R_\alpha(\rho), \nonumber
    \end{equation}
    where $R_\alpha(\rho) = \frac{1}{1-\alpha} \log \operatorname{Tr}(\rho^\alpha)$ is defined for $\alpha>0, \alpha \neq 1.$ 
    \end{lemma}
    \begin{remark}
        \begin{itemize}
            \item [1.]  Let us define a generalized cross entropy $C_\alpha(\rho,\sigma)$ as
            \[
            C_\alpha(\rho,\sigma) = \frac{\alpha}{1 - \alpha}\log \operatorname{Tr} (\rho \sigma^{\alpha - 1}) + \log \operatorname{Tr} (\sigma ^\alpha).
            \]
        Observe that, when $\rho=\sigma$, then
        \[
        C_\alpha(\rho,\rho) = R_\alpha(\rho).
        \]
        Thus, we have $S_\alpha(\rho||\sigma) = C_\alpha(\rho,\sigma) - R_\alpha(\rho)$ with $C_\alpha(\rho,\rho) = R_\alpha(\rho)$.

        \item[2.]  In particular when $\sigma$ is the maximally mixed stated, that is, $q_j = 1/n$ for all $j,$ then $S_{\alpha}(\rho || \sigma) = \log(n) - R_\alpha(\rho).$ This results to the well-known property of R\'enyi entropy $R_\alpha(\rho)\leq \log(n)$

        \item[3.]  
      Observe that $R_\alpha(\rho) \to S(\rho)$ as $\alpha \to 1,$ where $S(\rho)$ is the von Neumann entropy \eqref{vonNeumann}.
         \end{itemize}
    \end{remark}

Despite its connections to R\'enyi divergences, $S_{\alpha}(\rho \| \sigma)$ fails to satisfy several structural properties enjoyed by other popular divergences in quantum information theory, most notably the class of quantum $f$-divergences \eqref{standard_f}. As noted earlier, $S_{\alpha}$ lacks joint convexity and does not exhibit monotonicity with respect to the parameter $\alpha$. Moreover, for a fixed state $\rho$, the divergence $S_{\alpha}(\rho \| \sigma)$ does not preserve ordering in its second argument.

In contrast, the sandwiched R\'enyi divergence $D_{\alpha}^{*}(\rho \| \sigma)$ as in \eqref{Sandwiched} satisfies this monotonicity property; specifically,
\begin{equation}
D_{\alpha}^{*}(\rho \| \sigma_0) \leq D_{\alpha}^{*}(\rho \| \sigma)
\quad \text{whenever } \sigma_0 \geq \sigma,
\end{equation}
as shown in \cite{MullerMartin2013}. Umegaki's relative entropy \eqref{Umegaki} also obeys this inequality. Figure.~\ref{fig2} illustrates the contrasting behavior of $S_{\alpha}(\rho \| \sigma)$ and $\hat{D}_{\alpha}(\rho \| \sigma)$ across several representative scenarios. Furthermore, $S_{\alpha}(\rho \| \sigma)$ does not, in general, satisfy the data-processing inequality. We support this claim with the examples that follow. 

\begin{example}
    We fix $\alpha = 0.5$. Let $\rho$ and $\sigma$ be two density matrices in the system $\mathcal{H}_A$ with basis elements : $\mathcal{B}_A = \{|0\rangle, |1\rangle \}$. 

Let $\rho = \begin{pmatrix}
    0.85 & 0 \\ 0 & 0.15
\end{pmatrix}, \quad \sigma = \begin{pmatrix}
    0.25 & 0 \\ 0 & 0.75
\end{pmatrix}$.

Then $S_{.5}(\rho || \sigma) \approx 1\log(1.873) - 2\log(1.309) + \log(1.366) \approx 0.5782$.

Let $\Phi_1$ be a quantum channel from the system $\mathcal{H}_A$ to $\mathcal{H}_B,$ where $\mathcal{B}_B = \{|0\rangle, |1\rangle\}$ is the basis for $\mathcal{H}_B$. 

We define the quantum channel $\Phi_1$ as $\Phi_1(\rho) = \sum_{i=1}^4 K_{i} \rho K_{i}^{*},$ where $K_{1} = \sqrt{0.6}|0_A\rangle \langle 0_B|, K_{2} = \sqrt{0.05}|0_A\rangle \langle 1_B|, K_{3} = \sqrt{0.4}|1_A\rangle \langle 0_B|$ and $K_{4} = \sqrt{0.95}|1_A\rangle \langle 1_B|$.

This $\Phi_1$ is a well-defined quantum channel as $\sum_{i=1}^4 K_{i}^{*} K_{i}  = \mathcal{I}_2,$ where $\mathcal{I}_2$ is the identity operator of order $2$. 

After applying the channel on the original density matrices, we get 

$\Phi_1(\rho) = \begin{pmatrix}
    0.5175 & 0 \\ 0 & 0.4825
\end{pmatrix}$ and $\Phi_1(\sigma) = \begin{pmatrix}
    0.1875 & 0 \\ 0 & 0.8125
\end{pmatrix}.$

And we have $S_{.5}(\Phi_1(\rho) || \Phi_1(\sigma)) \approx 1\log(1.730) - 2\log(1.414) + \log(1.334) \approx 0.20664. $

Therefore $S_{\alpha}(\Phi_1(\rho) || \Phi_1(\sigma)) < S_{\alpha}(\rho ||\sigma)$ in this case, implying no violation of the data processing inequality.
\end{example}
    
\begin{example}
    We fix $\alpha = 2.$ Let $\rho$ and $\sigma$ be two density matrices in the system $\mathcal{H}_A$ with basis elements : $\mathcal{B}_A = \{|e_1\rangle, |e_2\rangle, |e_3\rangle \},$ where $|e_1\rangle = \begin{pmatrix}
    1 \\ 0 \\0
\end{pmatrix}, |e_2\rangle = \begin{pmatrix}
    0 \\ 1 \\0
\end{pmatrix},$ and $|e_3\rangle = \begin{pmatrix}
    0 \\ 0 \\1
\end{pmatrix}.$

Let $\rho = \begin{pmatrix}
0.5 & 0 & 0 \\
0 & 0.25 & 0 \\
0 & 0 & 0.25
\end{pmatrix}
$ and $\sigma = \begin{pmatrix}
0.7 & 0 & 0 \\
0 & 0.2 & 0 \\
0 & 0 & 0.1
\end{pmatrix}
.$ 

Then $S_{2}(\rho \| \sigma) = (-2) \log(0.425) + \log(0.375) + \log(0.54) \approx 0.1649.$

Let $\Phi_2$ be a quantum channel from the system $\mathcal{H}_A$ to $\mathcal{H}_B,$ where $\mathcal{B}_B = \{|0\rangle, |1\rangle\}$ is the basis for $\mathcal{H}_B$. 

We define the quantum channel $\Phi_2$ as $\Phi_2(\rho) = \sum_{i=1}^3 K_{i} \rho K_{i}^{*}$, where $K_{1} = |0\rangle \langle e_1|, K_{2} = |1\rangle \langle e_2|$ and $K_{3} = |1\rangle \langle e_3|$.

This $\Phi_2$ is a well-defined quantum channel as $\sum_{i=1}^3 K_{i}^{*} K_{i} = \mathcal{I}_3,$ where $\mathcal{I}_3$ is the identity operator of order $3$. 

After applying the channel on the original density matrices, we get 

$\Phi_2(\rho) = \begin{pmatrix}
    0.5 & 0 \\ 0 & 0.5
\end{pmatrix}$ and $\Phi_2(\sigma) = \begin{pmatrix}
    0.7 & 0 \\ 0 & 0.3
\end{pmatrix}.$

Therefore $S_{2}(\Phi_2(\rho) || \Phi_2(\sigma)) = (-2) \log(0.50) + \log(0.50) + \log(0.58) \approx 0.21412.$

Very clearly $S_{\alpha}(\Phi_2(\rho) || \Phi_2(\sigma)) > S_{\alpha}(\rho ||\sigma)$ in this case.

\end{example}

\begin{figure}
\centering
\includegraphics[width=0.9\textwidth]{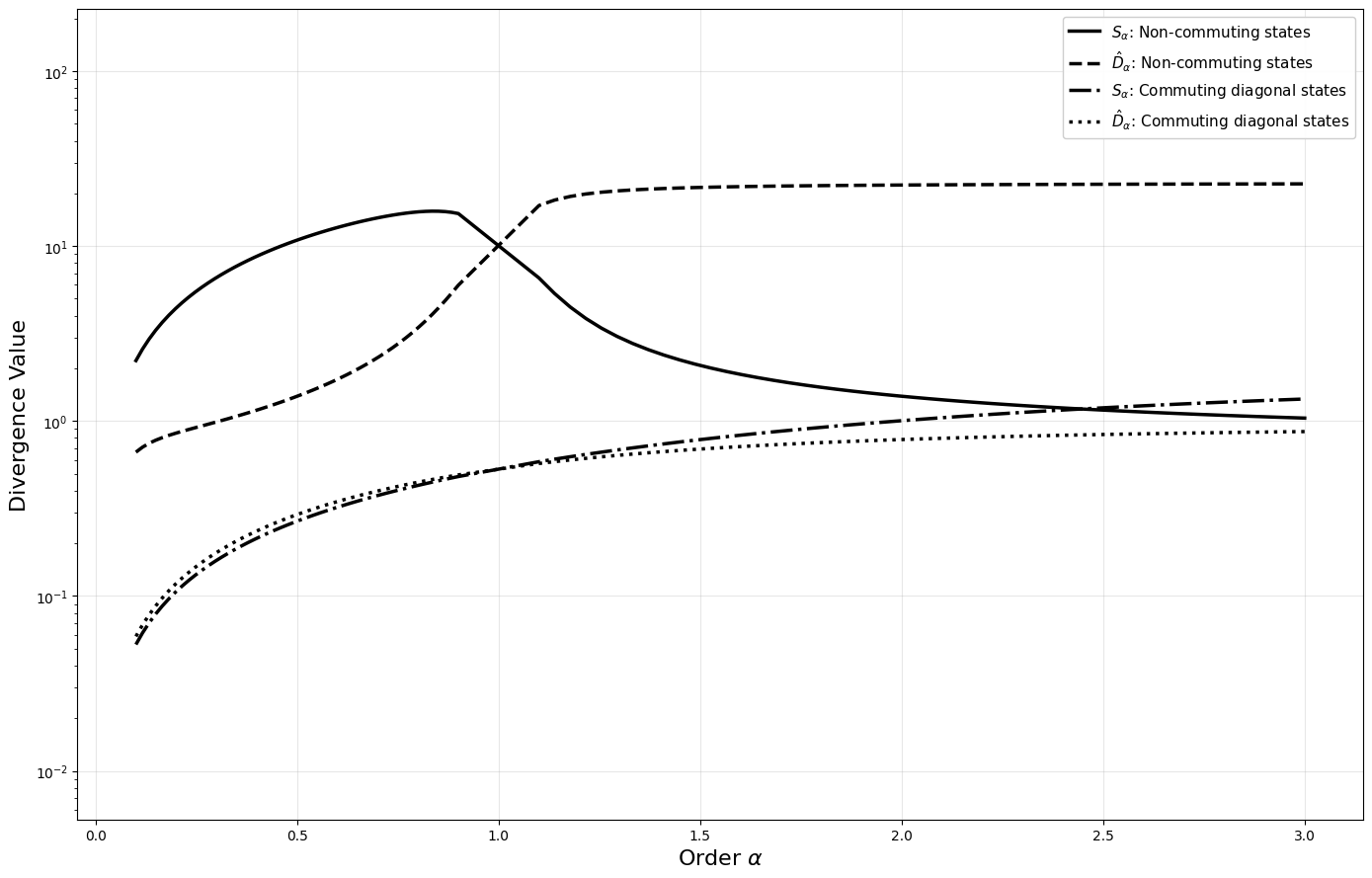}
\caption{The Quantum Relative $\alpha$-Entropy vs Petz-R\'enyi-$\alpha$-Relative Entropy as functions of the order $\alpha$.}
\label{fig2}
\end{figure}

We conclude this section with two more instances of special cases. First, we check the limiting behavior of $S_{\alpha}(\rho || \sigma)$ as $\alpha \to 0.$ Observe that 
\begin{eqnarray}
\label{limitatzero}
      \lim_{\alpha \to 0}  S_\alpha(\rho||\sigma)  =  0 - \log \operatorname{Tr}(\Pi_\rho) + \log \operatorname{Tr}(\Pi_\sigma) = \log \frac{\operatorname{Tr}(\Pi_\sigma)}{\operatorname{Tr}(\Pi_\rho)},
\end{eqnarray}
where $\Pi_\rho$ is the projection of the density matrix $\rho$ onto supp($\rho$). Datta in \cite{Datta2009} defined the min-relative entropy, $D_{min}(\rho||\sigma)$ of two density matrices $\rho$ and $\sigma$ as 
\begin{equation}
    D_{min}(\rho||\sigma) = - \log \operatorname{Tr}(\Pi_{\rho} \sigma). \nonumber
\end{equation}
And it is verified that 
\begin{equation}
    D_{min}(\rho||\sigma) = \lim_{\alpha \to 0}  \Hat{D}_\alpha(\rho||\sigma), \nonumber
\end{equation}
where $\Hat{D}_\alpha(\rho||\sigma)$ is as defined in \eqref{RenyiRelativeentropy}. But as shown in \eqref{limitatzero}, the limiting behavior of $S_{\alpha}(\rho || \sigma)$ at $\alpha = 0$ is quite different from $\Hat{D}_\alpha(\rho||\sigma)$, as no global comparison between $D_{min}(\rho||\sigma)$ and $\log \frac{\operatorname{Tr}(\Pi_\sigma)}{\operatorname{Tr}(\Pi_\rho)}$ can be obtained. However, in a specific setting, equality can be drawn. For instance, when $\sigma$ is maximally mixed,   
\begin{eqnarray}
      - \log \operatorname{Tr}(\Pi_{\rho}\sigma)  =  - \log \operatorname{Tr}\Big(\frac{\Pi_\rho}{n}\Big) = - \log \frac{\operatorname{Tr}(\Pi_\rho)}{n} = - \log \frac{\operatorname{Tr}(\Pi_\rho)}{\operatorname{Tr}(\Pi_\sigma)}. \nonumber
\end{eqnarray}
The converse of the above is also true, as stated in the following result. 

\begin{lemma}
\label{Lemma:minrelativerelation}
    $\lim\limits_{\alpha \to 0}  S_\alpha(\rho||\sigma) = D_{min}(\rho||\sigma)$ if and only if $\sigma$ is maximally mixed.
\end{lemma}
\begin{IEEEproof}
    For a maximally mixed state $\sigma,$ the statement is already shown to be true. We now prove the converse part of the lemma.

    Let $\lim\limits_{\alpha \to 0}  S_\alpha(\rho||\sigma) = D_{min}(\rho||\sigma),$ i.e., $\frac{\operatorname{Tr}(\Pi_\rho)}{\operatorname{Tr}(\Pi_\sigma)} = \operatorname{Tr}(\Pi_{\rho}\sigma).$ 

     Let $\operatorname{Tr}(\Pi_\sigma) = k$ for some real constant $k$.

Then $\operatorname{Tr}(\Pi_{\rho}\sigma) = \frac{1}{k} \operatorname{Tr}(\Pi_\rho)$.

For any pure density matrix $\hat{\rho}, \Pi_{\hat{\rho}} = |x \rangle \langle x|,$ for some orthonormal vector $|x \rangle \in \mathcal{H}.$

So, $\operatorname{Tr}(\Pi_{\hat{\rho}}) = 1$ and $\operatorname{Tr}(\Pi_{\hat{\rho}}\sigma) = \langle x|\sigma|x\rangle = \frac{1}{k}.$

This implies that $ \langle x|\sigma|x\rangle = \frac{1}{k} \quad \forall x$ with $||x||=1.$

Now let $\{|y_j \rangle\}_{i=1}^n$ be an orthonormal eigen-basis of $\sigma$.

Then the eigen values of $\sigma := q_j = \langle y_j|\sigma|y_j\rangle = \frac{1}{k} \quad \forall j$.

As $\sum_{j=1}^n q_j = 1$, we have $n \frac{1}{k} = 1 \implies q_j = \frac{1}{n} \quad \forall j$, where $n = dim(\mathcal{H})$.

This completes the proof.
\end{IEEEproof}

Lastly, for $\alpha = 2,$ we relate the quantum relative $\alpha$-entropy with the fidelity of two quantum states when the states are commutative. Using the definition \eqref{quantumjonesetal}, we get
\begin{eqnarray}
    S_2(\rho||\sigma)
    & = & -2 \log \operatorname{Tr} (\rho \sigma) + \log \operatorname{Tr} (\rho^2) + \log \operatorname{Tr} (\sigma^2) \nonumber \\
    & = & \log \Big[\frac{\operatorname{Tr}(\rho^2) \operatorname{Tr}(\sigma^2)}{\{\operatorname{Tr}(\rho \sigma)\}^2} \Big] \nonumber \\ 
    & = & \log \Big[\frac{\operatorname{Tr}(\rho^2) \operatorname{Tr}(\sigma^2)}{\{F(\rho, \sigma)\}^2} \Big], \nonumber 
\end{eqnarray}
where $F(\rho, \sigma) = \operatorname{Tr} [(\rho^{1/2} \sigma \rho^{1/2})^{1/2}]$ is defined as fidelity \cite{Jozsa94, Nielsen00, Datta15}.


\section{Nussbaum--Szko\l{}a Distributions and Quantum Relative $\alpha$-Entropy}
\label{5sec:NZDistributions}

To establish an exact correspondence between the classical relative $\alpha$-entropy \eqref{Jones} and its quantum analogue \eqref{quantumjonesetal}, we employ the Nussbaum-Szko\l{}a (NZ) distributions associated with a pair of quantum states. This construction enables us to represent the quantum relative $\alpha$-entropy $S_{\alpha}(\rho \| \sigma)$ as a classical relative $\alpha$-entropy evaluated on suitably defined probability measures.

Let $\rho$ and $\sigma$ be density matrices with spectral decompositions as in \eqref{decompositions}. The Nussbaum-Szko\l{}a distributions $P$ and $Q$ associated with $\rho$ and $\sigma$, respectively, are defined by
\begin{equation}
\label{distributioneqn}
P(i,j) = p_i |\langle x_i | y_j \rangle|^2,
\qquad
Q(i,j) = q_j |\langle x_i | y_j \rangle|^2.
\end{equation}

\begin{remark}
Since $\{y_j\}$ forms an orthonormal basis, $\sum_{j=1}^n |\langle x_i | y_j \rangle|^2 = 1$ for each $i$. Consequently,
\[
\sum_{i,j=1}^n P(i,j)
=
\sum_{i=1}^n p_i
=
1,
\]
and $P(i,j) \ge 0$ for all $i,j$. Hence $P$ is a probability distribution. The same argument applies to $Q$.
\end{remark}

For $\alpha > 0$, define the associated $\alpha$-escort NZ distributions
\begin{equation}
\label{distributioneqnmodi}
P^{(\alpha)}(i,j)
=
\frac{p_i^{\alpha}}{\operatorname{Tr}(\rho^{\alpha})}
|\langle x_i | y_j \rangle|^2,
\qquad
Q^{(\alpha)}(i,j)
=
\frac{q_j^{\alpha}}{\operatorname{Tr}(\sigma^{\alpha})}
|\langle x_i | y_j \rangle|^2.
\end{equation}

\begin{lemma}
\label{Lemma3}
For $\alpha > 0$, the classical $f$-divergence between the escort NZ distributions satisfies
\begin{equation}
\label{fdivNZ}
D_f\!\left(P^{(\alpha)} \| Q^{(\alpha)}\right)
=
\sum_{\substack{i,j:\\ \langle x_i | y_j \rangle \neq 0}}
f\!\left(\frac{p_i^{(\alpha)}}{q_j^{(\alpha)}}\right)
q_j^{(\alpha)}
|\langle x_i | y_j \rangle|^2,
\end{equation}
where
\[
p_i^{(\alpha)} = \frac{p_i^{\alpha}}{\operatorname{Tr}(\rho^{\alpha})},
\qquad
q_j^{(\alpha)} = \frac{q_j^{\alpha}}{\operatorname{Tr}(\sigma^{\alpha})}.
\]
\end{lemma}

The expression in \eqref{fdivNZ} follows directly by substituting \eqref{distributioneqnmodi} into the classical definition \eqref{f-divergence}; see also \cite{Androulakis2024} for related formulations.

\begin{remark}
\label{Remark:fconnection}
Let $f(x) = \mathrm{sgn}\!\left(\frac{1-\alpha}{\alpha}\right)(x^{1/\alpha} - 1)$ for $\alpha > 0$. Then
\[
f\!\left(\frac{p_i^{(\alpha)}}{q_j^{(\alpha)}}\right)
=
\mathrm{sgn}\!\left(\frac{1-\alpha}{\alpha}\right)
\left[
\left(\frac{p_i}{q_j}\right)
\left(
\frac{\operatorname{Tr} \sigma^{\alpha}}
{\operatorname{Tr} \rho^{\alpha}}
\right)^{1/\alpha}
- 1
\right].
\]
Substituting into \eqref{fdivNZ} yields an explicit expression for
$D_f(P^{(\alpha)} \| Q^{(\alpha)})$ in terms of the eigenvalues of $\rho$ and $\sigma$.
\end{remark}

Following \cite{KumarS15J1}, the classical relative $\alpha$-entropy is defined by
\begin{equation}
\label{jonesin2015}
J_{\alpha}(P \| Q)
=
\frac{\alpha}{1 - \alpha}
\log
\left[
\mathrm{sgn}\!\left(\frac{1-\alpha}{\alpha}\right)
D_f\!\left(P^{(\alpha)} \| Q^{(\alpha)}\right)
+ 1
\right],
\end{equation}
where $P^{(\alpha)}$ and $Q^{(\alpha)}$ denote the corresponding $\alpha$-escort measures.

Substituting the expression obtained in Remark~\ref{Remark:fconnection} into \eqref{jonesin2015} yields
\[
J_{\alpha}(P \| Q)
=
\frac{\alpha}{1 - \alpha}
\log
\left[
\sum_{\substack{i,j:\\ \langle x_i | y_j \rangle \neq 0}}
p_i q_j^{\alpha-1}
(\operatorname{Tr} \rho^{\alpha})^{-1/\alpha}
(\operatorname{Tr} \sigma^{\alpha})^{(1-\alpha)/\alpha}
|\langle x_i | y_j \rangle|^2
\right],
\]
which coincides with \eqref{quantumjones2}. We thus obtain the following result.

\begin{theorem}
\label{1:Theorem7}
Let $\rho$ and $\sigma$ be density matrices with spectral decompositions as in \eqref{decompositions}, and let $P$ and $Q$ denote their associated NZ-distributions. Then
\[
S_{\alpha}(\rho \| \sigma)
=
J_{\alpha}(P \| Q),
\]
where $S_{\alpha}(\rho \| \sigma)$ denotes the quantum relative $\alpha$-entropy and $J_{\alpha}(P \| Q)$ the classical relative $\alpha$-entropy.
\end{theorem}
\begin{IEEEproof}
    From \eqref{fdivNZ}, using the function from Remark~\ref{Remark:fconnection}, we have
    \begin{align*}
        & D_f\!\left(P^{(\alpha)} \| Q^{(\alpha)}\right) \\
        & =  \sum_{\substack{i,j:\\ \langle x_i | y_j \rangle \neq 0}} \mathrm{sgn}\!\left(\frac{1-\alpha}{\alpha}\right) \left[ \left(\frac{p_i}{q_j}\right) \left(\frac{\operatorname{Tr} \sigma^{\alpha}} {\operatorname{Tr} \rho^{\alpha}} \right)^{1/\alpha} - 1 \right] q_j^{(\alpha)} |\langle x_i | y_j \rangle|^2 \nonumber \\
        & = \mathrm{sgn}\!\left(\frac{1-\alpha}{\alpha}\right) \left[ \sum_{\substack{i,j:\\ \langle x_i | y_j \rangle \neq 0}} \left(\frac{p_i}{q_j}\right) \left(\frac{\operatorname{Tr} \sigma^{\alpha}} {\operatorname{Tr} \rho^{\alpha}} \right)^{1/\alpha} \left( \frac{q_j^{\alpha}}{\operatorname{Tr} \sigma^{\alpha}}\right) |\langle x_i | y_j \rangle|^2 - \sum_{\substack{i,j:\\ \langle x_i | y_j \rangle \neq 0}} \frac{q_j^{\alpha}}{\operatorname{Tr} \sigma^{\alpha}} |\langle x_i | y_j \rangle|^2 \right] \\ \nonumber
        & = \mathrm{sgn}\!\left(\frac{1-\alpha}{\alpha}\right) \left[\sum_{\substack{i,j:\\ \langle x_i | y_j \rangle \neq 0}} p_i q_j^{\alpha-1} \left(\operatorname{Tr} \sigma^{\alpha}\right)^{1/\alpha-1} \left(\operatorname{Tr} \rho^{\alpha}\right)^{-1/\alpha} |\langle x_i | y_j \rangle|^2 - \sum_{\substack{j:\\ \langle x_i | y_j \rangle \neq 0}} \frac{q_j^{\alpha}}{\operatorname{Tr} \sigma^{\alpha}} \right] ,\nonumber
\end{align*}
with
\begin{equation*}
    \sum_{\substack{j:\\ \langle x_i | y_j \rangle \neq 0}} \frac{q_j^{\alpha}}{\operatorname{Tr} \sigma^{\alpha}} = 1.
\end{equation*}
This implies 
\begin{equation*}
      \mathrm{sgn}\!\left(\frac{1-\alpha}{\alpha}\right) D_f\!\left(P^{(\alpha)} \| Q^{(\alpha)}\right) + 1 = \sum_{\substack{i,j:\\ \langle x_i | y_j \rangle \neq 0}} p_i q_j^{\alpha-1} \left(\operatorname{Tr} \sigma^{\alpha}\right)^{1/\alpha-1} \left(\operatorname{Tr} \rho^{\alpha}\right)^{-1/\alpha} |\langle x_i | y_j \rangle|^2.
\end{equation*}

And finally, following \eqref{jonesin2015} we have,
\begin{eqnarray}
    J_{\alpha}(P \| Q)
    & = &\frac{\alpha}{1 - \alpha} \log \sum_{\substack{i,j:\\ \langle x_i | y_j \rangle \neq 0}} \left[p_i q_j^{\alpha-1} \left(\operatorname{Tr} \sigma^{\alpha}\right)^{1/\alpha-1} \left(\operatorname{Tr} \rho^{\alpha}\right)^{-1/\alpha} |\langle x_i | y_j \rangle|^2 \right] \nonumber \\
    & = &\frac{\alpha}{1 - \alpha} \log \sum_{\substack{i,j:\\ \langle x_i | y_j \rangle \neq 0}} p_i q_j^{\alpha-1} |\langle x_i | y_j \rangle|^2 - \frac{1}{1 - \alpha} \log \sum_{i=1}^n p_i^{\alpha} + \log \sum_{j=1}^n q_j^{\alpha}, \nonumber
\end{eqnarray}
which is equivalent to the expression \eqref{quantumjones2} of the quantum relative $\alpha$-entropy.

This completes the proof.  
\end{IEEEproof}

\begin{remark}
Lemma~\ref{Lemma:positivity} also follows directly from Theorem~\ref{1:Theorem7} together with \cite[Lemma~2]{KumarS15J1}.
\end{remark}
\section{Quantum Density Power Divergence}
\label{6sec:QuantumDPD}

In this section, we introduce another divergence measure between two density matrices, termed the \emph{Quantum Density Power Divergence}. The motivation for considering this alternative construction stems from the broader observation that different generalizations of the KL divergence capture different structural aspects of quantum distinguishability. While several extensions preserve properties such as data-processing inequality or joint convexity, others arise naturally from statistical considerations. The divergence introduced below is motivated by the latter perspective and is inspired by the classical density power divergence \cite{BasuHHJ98J}.

\begin{definition}
\label{Defn:QDPD}
Let $\alpha > 0$ with $\alpha \neq 1$.  
For two density matrices $\rho$ and $\sigma$, the Quantum Density Power Divergence is defined as
\begin{eqnarray}
\label{QDPD}
\overline{S}_{\alpha}(\rho \| \sigma)
& = &
\frac{\alpha}{1 - \alpha} \operatorname{Tr} (\rho \sigma^{\alpha - 1})
- \frac{1}{1 - \alpha} \operatorname{Tr} (\rho^\alpha)
+ \operatorname{Tr} (\sigma^\alpha),
\end{eqnarray}
whenever $\mathrm{supp}[\rho] \subseteq \mathrm{supp}[\sigma]$. Otherwise,
$\overline{S}_{\alpha}(\rho \| \sigma) = +\infty$.
\end{definition}

Assuming that $\rho$ and $\sigma$ admit the spectral decompositions in \eqref{decompositions}, the divergence in \eqref{QDPD} can be written as
\begin{eqnarray}
\label{QDPD2}
\overline{S}_{\alpha}(\rho \| \sigma)
& = &
\frac{\alpha}{1 - \alpha}
\sum_{i,j=1}^n p_i q_j^{\alpha - 1}
|\langle x_i | y_j \rangle|^2
- \frac{1}{1 - \alpha} \sum_{i=1}^n p_i^{\alpha}
+ \sum_{j=1}^n q_j^{\alpha}.
\nonumber
\end{eqnarray}

The definition is well posed for positive semi-definite complex matrices. The divergence satisfies non-negativity,
\[
\overline{S}_{\alpha}(\rho \| \sigma) \ge 0,
\]
with equality if and only if $\rho = \sigma$. It is invariant under unitary conjugation, but it is not additive under tensor products. Similar to the quantum relative $\alpha$-entropy $S_{\alpha}(\rho \| \sigma)$, it fails to be jointly convex, although it remains convex in the first argument for all $\alpha>0$. In general, it does not satisfy the data-processing inequality. Unlike $S_{\alpha}(\rho \| \sigma)$, however, it is not invariant under positive scalar multiplication (see Lemma~\ref{Lemma:Constants}). Its structural connections with standard quantum information measures are therefore more limited.

The definition in \eqref{QDPD} is motivated by the classical density-power divergence
\[
\mathcal{B}_{\alpha}(p \| q)
=
\frac{\alpha}{1 - \alpha} \sum_{i,j=1}^n p_i q_j^{\alpha - 1}
- \frac{1}{1 - \alpha} \sum_{i=1}^n p_i^{\alpha}
+ \sum_{j=1}^n q_j^{\alpha},
\]
where $p=(p_i)_{i=1}^n$ and $q=(q_i)_{i=1}^n$ are two probability distributions. 
This lies outside the class of classical $f$-divergences \eqref{f-divergence}, for probability distributions $p$ and $q$ with common support. This, however, belongs to a bigger class called Bregman divergences whose projections are uniquely characterized by transitive rules \cite[Ex. 3]{Kanamori14}.

As $\alpha \to 1$, one recovers
\[
\overline{S}_{\alpha}(\rho \| \sigma) \to U(\rho \| \sigma),
\qquad
\mathcal{B}_{\alpha}(p \| q) \to \mathrm{KL}(p \| q),
\]
where $\mathrm{KL}(p \| q)$ denotes the KL divergence defined in \eqref{KL}. In the classical literature, this generalization of the KL divergence is known as the density power divergence \cite{Basu98, Kanamori14, Naudts04}. It has been studied extensively in the context of robust parameter estimation \cite{Basu98, JonesHHB01J}, as well as in hypothesis testing \cite{Basu13} and regression and multivariate modeling \cite{Riani20}.


\section{Summary and Conclusion}
\label{7sec:summary}
In this work, we introduced a new class of quantum divergences, termed the \emph{quantum relative $\alpha$-entropy} $S_{\alpha}(\rho \| \sigma)$. This generalizes Umegaki’s relative entropy while lying strictly outside the class of standard quantum $f$-divergences. We identified precise support conditions on density operators that ensure finiteness of $S_{\alpha}(\rho \| \sigma)$ and established a collection of structural properties that are essential for a meaningful quantum divergence.
We proved that $S_{\alpha}(\rho \| \sigma)$ is additive under tensor products and invariant under unitary conjugations, ensuring consistency under composition of independent quantum systems and basis transformations. A distinctive feature of the proposed divergence is its invariance under independent positive rescaling of both arguments, a property not shared by $f$-divergences. This invariance highlights a fundamentally new structural behavior, showing that $S_{\alpha}(\rho \| \sigma)$ depends only on the intrinsic spectral and coherence structure of the states rather than on their absolute normalization.

Unlike conventional quantum divergences, $S_{\alpha}(\rho \| \sigma)$ is not jointly convex. This motivated the development of a generalized convexity framework adapted to its multiplicative structure. Our analysis distinguishes the space of density matrices from the classical probability simplex, where such generalized convexity is always valid. For the quantum state space, it holds only on a restricted but well-defined subclass of density operators. This observation naturally suggests an alternative route toward a generalized data processing inequality for $S_{\alpha}(\rho \| \sigma)$, a direction that we plan to pursue in future work.

We further positioned $S_{\alpha}(\rho \| \sigma)$ within the broader landscape of R\'enyi-type quantum divergences. Under suitable invertible transformations, we showed that it recovers the Petz–R\'enyi relative entropy, and we established necessary and sufficient conditions relating it to the min-relative entropy. These results clarify both the connections and the essential differences between the proposed divergence and existing quantum information measures.

A central result of this paper is the reduction of quantum relative $\alpha$-entropy to its classical counterpart via the Nussbaum–Szko\l{}a distributions. We proved that $S_{\alpha}(\rho \| \sigma)$ can be expressed exactly as a classical relative-$\alpha$-entropy between probability distributions derived from the spectral data of the relative modular operator. This representation demonstrates that quantum distinguishability, even in the noncommuting case, admits a faithful classical description based on physically realizable measurement statistics, with the commuting case recovered as a special instance.

Finally, motivated by the structural form of relative $\alpha$-entropy, we introduced a new Bregman-type quantum divergence inspired by the classical density power divergence. We analyzed its relationship with $S_{\alpha}(\rho \| \sigma)$ and highlighted both shared structural features and key differences. Together, these results establish quantum relative $\alpha$-entropy as a unifying framework connecting R\'enyi-type divergences, Bregman divergences, and robust information measures, opening new directions for quantum information geometry and quantum statistical inference.

\bibliographystyle{IEEEtranS}
\bibliography{ReferencesQIT}

@ARTICLE{BasuHHJ98J,
  author  = {A. Basu and I. R. Harris and N. L. Hjort and M. C. Jones},
  title   = {Robust and efficient estimation by minimizing a density power divergence},
  journal = {Biometrika}, 
  volume  = {85},
  year    = {1998},
  pages   = {549--559}
}

@ARTICLE{KumarM20J,
  author  = {M. A. Kumar and K. V. Mishra},
  title   = {Cram\'er–{R}ao lower bounds arising from generalized {C}sisz\'ar divergences},
  journal = {Info. Geo.}, 
  volume  = {3},
  year    = {2020},
  pages   = {33--59}
}

@ARTICLE{GayenK21J,
  author  = {A. Gayen and M. A. Kumar} ,
  title   = {Projection theorems and estimating equations for power-law models},
  journal = {Journal of Multivariate Analysis}, 
  volume  = {184},
  year    = {2021},
  pages   = {104734}
}

@ARTICLE{KumarS15J1,
  author  = {M. A. Kumar and R. Sundaresan},
  title   = {Minimization problems based on relative $\alpha$-entropy~I: Forward Projection},
  journal = {IEEE Transactions on Information Theory}, 
  volume  = {61},
  year    = {2015},
  pages   = {5063--5080}
}

@ARTICLE{KumarS15J2,
  author  = {M. A. Kumar and R. Sundaresan},
  title   = {Minimization problems based on relative $\alpha$-entropy~II: Reverse Projection},
  journal = {IEEE Transactions on Information Theory}, 
  volume  = {61},
  year    = {2015},
  pages   = {5081--5095}
}

@ARTICLE{KumarS16J,
  author  = {M. A. Kumar and I. Sason},
  title   = {Projection theorems for the R\'enyi divergence on alpha-convex sets},
  journal = {IEEE Transactions on Information Theory}, 
  volume  = {62},
  year    = {2016},
  pages   = {4924--4935}
}

@ARTICLE{GayenRG24J,
  author = {A. Gayen and S. Roy and A. K. Gangopadhyay},
  title = {A unified approach to the Pythagorean identity and projection theorem for a class of divergences based on M-estimations},
  journal = {Statistics},
  volume = {58},
  year = {2024},
  pages = {842--880}
}

@ARTICLE{JonesHHB01J,
  author  = {M. C. Jones and N. L. Hjort and I. R. Harris and A. Basu},
  title   = {A comparison of related density based minimum divergence estimators},
  journal = {Biometrika}, 
  volume  = {88},
  year    = {2001},
  pages   = {865--873}
}

@ARTICLE{GayenK23J,
  author  = {A. Gayen and M. A. Kumar},
  title   = {Generalized {F}isher-{D}armois-{K}oopman-{P}itman Theorem and {R}ao-{B}lackwell Type Estimators for Power-Law Distributions},
  journal = {IEEE Transactions on Information Theory},
  volume = {69},
  year = {2023},
  pages = {7565--7583}
}

@ARTICLE{BasuBC97J,
  author  = {A. Basu and S. Basu and G. Chaudhury},
  title   = {Robust minimum divergence procedures for count data models},
  journal = {Sankhya: The Indian Journal of Statistic}, 
  volume  = {59},
  year    = {1997},
  pages   = {11--27}
}

@ARTICLE{Umegaki62,
  author = {H. Umegaki},
  title = {{Conditional expectation in an operator algebra. IV. Entropy and information}},
  journal = {Kodai Mathematical Seminar Reports},
  volume = {14},
  year = {1962},
  pages = {59--85},
}

@ARTICLE{OgawaNagaoka00,
  author = {T. Ogawa and H. Nagaoka},
  title = {Strong converse and Stein's lemma in quantum hypothesis testing},
  journal = {IEEE Transactions on Information Theory}, 
  volume = {46}, 
  year = {2000},
  pages = {2428-2433},
}

@ARTICLE{Kullback51,
  author = {S. Kullback and R. A. Leibler},
  title = {{On Information and Sufficiency}}, 
  journal = {The Annals of Mathematical Statistics},
  volume = {22},
  year = {1951},
  pages = {79--86},
}

@BOOK{Pardo06B,
  author = {L. Pardo},
  title  = {Statistical Inference Based on Divergence Measures},
  publisher = {Chapman \& Hall/CRC, Taylor	and Francis group, Boca Raton, Florida, USA},
  year = {2006}
}

@ARTICLE{CressieR84J,
  author  = {N. Cressie and  T. R. C. Read},
  title   = {Multinomial Goodness-of-Fit Tests},
  journal = {J. R. Stat. Soc. Ser. B. Stat. Methodol.}, 
  volume  = {46},
  year    = {1984},
  pages   = {440--464}
}

@ARTICLE{Tsallis88J,
  author  = {C. Tsallis},
  title   = {Possible generalization of Bolzmann-Gibbs statistics},
  journal = {J. Stat. Phys.}, 
  volume  = {52},
  year    = {1988},
  pages   = {479--487}
}

@ARTICLE{TsallisMP98J,
  author  = {C. Tsallis and R. S. Mendes and A. R. Plastino},
  title   = {The role of constraints within generalized non-extensive statistics},
  journal = {Phys. A.}, 
  volume  = {261},
  year    = {1998},
  pages   = {534--554}
}

@ARTICLE{Araki1975,
  author={H. Araki},
  title={Relative Entropy of States of von Neumann Algebras},
  journal={Publications of The Research Institute for Mathematical Sciences},
  volume={11},
  year={1975},
  pages={809-833}
}

@ARTICLE{Araki2005,
  author={H. Araki},
  title={Relative Entropy for States of von Neumann Algebras II},
  journal={Publications of The Research Institute for Mathematical Sciences},
  year={2005}
}

@ARTICLE{MullerMartin2013,
    author = {M. Müller-Lennert and F. Dupuis and O. Szehr and S. Fehr and M. Tomamichel},
    title = {On quantum Rényi entropies: A new generalization and some properties},
    journal = {Journal of Mathematical Physics},
    volume = {54},
    year = {2013},
    pages = {122203}
}

@ARTICLE{Nussbaum2009,
  author = {M. Nussbaum and A. Szko{\l}a},
  title = {The Chernoff lower bound for symmetric quantum hypothesis testing},
  journal = {The Annals of Statistics},
  volume = {37},
  year = {2009},
  pages = {}	
}

@ARTICLE{Hiai2017,
  author = {F. Hiai and M. Mosonyi},
  title = {Different quantum f-divergences and the reversibility of quantum operations},
  journal = {Reviews in Mathematical Physics},
  volume = {29},
  year = {2017},
  pages = {1750023}
}

@ARTICLE{Hiai2018,
  author = {F. Hiai},
  title = {Quantum $f$-divergences in von Neumann algebras. I. Standard $f$-divergences},
  journal = {Journal of Mathematical Physics},
  volume = {59},
  year = {2018},
  pages = {102202}	
}

@ARTICLE{Csiszar1967,
  author={I. Csiszár},
  title={Information-type measures of difference of probability distributions and indirect observation},
  journal={Studia Scientiarum Mathematicarum Hungarica},
  volume={2},
  year={1967},
  pages={229-318}
}

@ARTICLE{Ervan2014,
  author={T. V. Erven and P. Harremos},
  title={Rényi Divergence and Kullback-Leibler Divergence},
  journal={IEEE Transactions on Information Theory}, 
  volume={60},
  year={2014},
  pages={3797-3820}
}

@ARTICLE{Sundaresan2002,
  author={R. Sundaresan},
  title={A measure of discrimination and its geometric properties},
  journal={Proceedings IEEE International Symposium on Information Theory},
  year={2002},
  pages={264-},
}

@ARTICLE{Sundaresan2007,
  author = {R. Sundaresan},
  title = {Guessing Under Source Uncertainty},
  journal = {IEEE Transactions on Information Theory},
  volume = {53},
  year = {2007},
  pages = {269–287}
}

@ARTICLE{Androulakis2024,
  author = {G. Androulakis and T. C. John},
  title = {Quantum f-divergences via Nussbaum–Szkoła distributions and applications to f-divergence inequalities},
  journal = {Reviews in Mathematical Physics},
  volume = {36},
  year = {2024},
  pages = {2360002}
}

@ARTICLE{Datta2009,
  author = {N. Datta},
  title = {Min- and max-relative entropies and a new entanglement monotone},
  journal = {IEEE Transactions on Information Theory},
  volume = {55},
  year = {2009},
  pages = {2816–2826}
}

@ARTICLE{Bluhm2022,
  author={A. Bluhm and {\'A}. Capel and P. Gondolf and A. P. Hern{\'a}ndez},
  title={Continuity of Quantum Entropic Quantities via Almost Convexity},
  journal={IEEE Transactions on Information Theory},
  volume={69},
  year={2022},
  pages={5869-5901}
}

@ARTICLE{Ohya2010,
  author = {M. Ohya and N. Watanabe},
  title = {Quantum Entropy and Its Applications to Quantum Communication and Statistical Physics},
  journal = {Entropy},
  volume = {12},
  year = {2010},
 pages = {1194--1245}
}

@ARTICLE{Petz1986,
  author = {D. Petz},
  title = {Quasi-entropies for finite quantum systems},
  journal = {Reports on Mathematical Physics},
  volume = {23},
  year = {1986},
pages = {57-65}
}

@ARTICLE{Uhlmann77,
  author = {A. Uhlmann},
  title = {Relative entropy and the Wigner-Yanase-Dyson-Lieb concavity in an interpolation theory},
  journal = {Communications in Mathematical Physics},
  volume = {54},
  year = {1977},
pages = {21–32}
}

@ARTICLE{Basu98,
  author = {A. Basu and I. R. Harris and N. L. Hjort and M.C. Jones},
  title = {Robust and efficient estimation by minimising a density power divergence},
  journal = {Biometrika},
  volume = {85},
  year={1998},
  pages={549-559}
}

@ARTICLE{Kanamori14,
  author = {T. Kanamori},
  title = {Scale-Invariant Divergences for Density Functions},
  journal = {Entropy},
  volume = {16},
  year = {2014},
  pages = {2611--2628}
}

@ARTICLE{Naudts04,
  author={J. Naudts},
  title={Estimators, escort probabilities, and $\phi$-exponential families in statistical physics},
  journal = {Journal of Inequalities in Pure \& Applied Mathematics},
  volume = {5},
  year = {2004},
  pages = {102} 
}

@ARTICLE{Basu13,
  author = {A. Basu, A. Mandal, N. Martin and L. Pardo},
  title = {Testing statistical hypotheses based on the density power divergence},
  journal = {Annals of the Institute of Statistical Mathematics},
  volume = {65},
  year = {2013},
  pages = {319-348}
}

@ARTICLE{Riani20,
  author = {M. Riani,A. C. Atkinson, A. Corbellini and D. Perrotta},
  title = {Robust Regression with Density Power Divergence: Theory, Comparisons, and Data Analysis},
  journal = {Entropy},
  volume = {22},
  year = {2020},
  pages = {1099-4300}
}

@ARTICLE{Datta15,
  author = {K. M. R. Audenaert and N. Datta},
  title = {$\alpha-z$-relative Renyi entropies},
  journal = {Journal of Mathematical Physics},
  Volume = {56},
  year = {2015},
  pages = {022202}
}

@ARTICLE{Temme10,
    author = {K. Temme and M. J. Kastoryano and M. B. Ruskai and M. M. Wolf and F. Verstraete},
    title = {The $\chi^2$-divergence and mixing times of quantum Markov processes},
    journal = {Journal of Mathematical Physics},
    volume = {51},
    year = {2010},
    pages = {122201}
}

@ARTICLE{Osaka25,
  author = {H. Osaka and H. Shudo},
  title = {Generalized Quantum Hellinger Divergences Generated by Monotone Functions},
  journal = {Open Systems \& Information Dynamics},
  volume = {32},
  year = {2025},
  pages = {2550013}
}

@ARTICLE{Jozsa94,
  author = {R. Jozsa},
  title = {Fidelity for Mixed Quantum States},
  journal = {Journal of Modern Optics},
  volume = {41},
  year = {1994},
  pages = {2315--2323}
}

@CONFERENCE{Renyi61J,
  author  = {A. R\'{e}nyi},
  title   = {On measures of entropy and information},
  booktitle = {Proceedings of 4th Berkeley Symposium on Mathematical Statistics and Probability, Berkeley, California, USA}, 
  year    = {1961},
  pages   = {547--561}
}

@BOOK{Nielsen00,
  author = {M. A. Nielsen and I. L. Chuang},
  title = {Quantum Computation and Quantum Information},
  address = {Cambridge},
  publisher = {Cambridge University Press},
  year = {2000}
}

@BOOK{Marcol2015,
  author = {M. Tomamichel},
  title = {Quantum Information Processing with Finite Resources -- Mathematical Foundations},
  address = {Cham},
  publisher = {Springer Cham},
  year = {2015}
}

@BOOK{Denes2007,
  author = {D. Petz},
  title = {Quantum information theory and quantum statistics},
  address = {Heidelberg},
  publisher = {Springer},
  year = {2007}
}

@BOOK{BasuSP11B,
  author = {A. Basu and H. Shioya and C. Park},
  title  = {Statistical Inference: The Minimum Distance Approach},
  publisher = {Chapman \& Hall/ CRC Monographs on Statistics and Applied Probability 120},
  year = {2011}
}

@BOOK{CsiszarS04B,
  author = {I. Csisz\'ar and P. C. Shields},
  title  = {Information Theory and Statistics: A Tutorial},
  address = {Hanover},
  publisher = {Foundations and Trends in Communications and Information Theory},
  year = {2004}
}

@BOOK{zhang11,
  author = {F. Zhang},
  title = {Matrix Theory: Basic Results and Techniques},
  address = {New York},
  publisher = {Springer},
  year = {2011}
}

@BOOK{Parthasarathy06,
  author = {K. R. Parthasarathy},
  title = {Lectures on Quantum computation, quantum error correcting codes and information theory},
  address = {New Delhi, India},
  publisher = {Narosa Pub.},
  year = {2006}
}

\end{document}